\documentclass[journal]{IEEEtran}
\usepackage{soul, color, xcolor} % 用于控制高亮的包
\usepackage{multirow} % 用于控制表格多行的包
\usepackage{amsmath,amsfonts}
\usepackage{algorithmic}
\usepackage{algorithm}
\usepackage{array}
\usepackage{bbm}
\ifCLASSOPTIONcompsoc
    \usepackage[caption=false, font=normalsize, labelfont=sf, textfont=sf]{subfig}
\else
\usepackage[caption=false, font=footnotesize]{subfig}
\fi
\usepackage{textcomp}
\usepackage{stfloats}
\usepackage{url}
\usepackage{verbatim}
\usepackage{graphicx}
\usepackage{epstopdf}
\graphicspath{{Figures/}}
\usepackage{cite}
\usepackage{amssymb} %\triangleq

\newenvironment{proof}{{\indent \it Proof:}}{\hfill $\blacksquare$\par}
\newenvironment{AppendixProof}{}{\hfill $\blacksquare$\par}
\newtheorem{theorem}{Theorem}

\begin{document}
\title{Performance Analysis of Integrated Sensing and Communication Networks with Blockage Effects}

\author{Zezhong Sun, Shi Yan,~\IEEEmembership{Member,~IEEE,} Ning Jiang, Jiaen Zhou, and Mugen Peng,~\IEEEmembership{Fellow,~IEEE}
    % <-this % stops a space
    \thanks{This work was supported by National Key R\&D Program of China (No. 2021YFB2900200). National Natural Science Foundation of China under(62371067, U21A20444, 61921003, 61831002), Beijing Natural Science Foundation under Grant L223007. }% <-this % stops a space
    \thanks{Copyright (c) 2024 IEEE. Personal use of this material is permitted. However, permission to use this material for any other purposes must be obtained from the IEEE by sending a request to pubs-permissions@ieee.org.}
    \thanks{Zezhong Sun (szz@bupt.edu.cn), Shi Yan (yanshi01@bupt.edu.cn), Ning Jiang (creativejn@bupt.edu.cn), Jiaen Zhou (circularje@bupt.edu.cn), and Mugen Peng (pmg@bupt.edu.cn) are with the State Key Laboratory of Networking and Switching Technology, Beijing University of Posts and Telecommunications, Beijing 100876, China. {\textbf{(Corresponding author: Shi Yan.)}}}
}

% The paper headers
\markboth{IEEE TRANSACTIONS ON VEHICULAR TECHNOLOGY,~VOL.~XX, NO.~X, XXXXXX~2024}%
{Zezhong Sun \MakeLowercase{\textit{et al.}}: Performance Analysis in Integrated Sensing and Communication (ISAC) Networks: A Stochastic Geometry Approach}

\IEEEpubid{\begin{minipage}{\textwidth}\ \\[12pt]\centering
        0000-0000~\copyright~2024 IEEE. Personal use is permitted, but republication/redistribution requires IEEE permission.\\
        See https://www.ieee.org/publications/rights/index.html for more information.
    \end{minipage}
}
% Remember, if you use this you must call \IEEEpubidadjcol in the second column for its text to clear the IEEEpubid mark.
\IEEEpubidadjcol

\maketitle

\begin{abstract}
    Communication-sensing integration represents an up-and-coming area of research, enabling wireless networks to simultaneously perform communication and sensing tasks. However, in urban cellular networks, the blockage of buildings results in a complex signal propagation environment, affecting the performance analysis of integrated sensing and communication (ISAC) networks. To overcome this obstacle, this paper constructs a comprehensive framework considering building blockage and employs a distance-correlated blockage model to analyze interference from line of sight (LoS), non-line of sight (NLoS), and target reflection cascading (TRC) links. Using stochastic geometric theory, expressions for signal-to-interference-plus-noise ratio (SINR) and coverage probability for communication and sensing in the presence of blockage are derived, allowing for a comprehensive comparison under the same parameters. The research findings indicate that blockage can positively impact coverage, especially in enhancing communication performance. The analysis also suggests that there exists an optimal base station (BS) density when blockage is of the same order of magnitude as the BS density, maximizing communication or sensing coverage probability.
\end{abstract}


\section{Introduction}
\subsection{Background and Motivations}
\IEEEPARstart{I}{ntegrated} sensing and communication (ISAC) is a potential key technology for beyond-fifth-generation (B5G) and sixth-generation (6G) wireless communications and has received wide attention from academia and industry \cite{liuIntegratedSensingCommunications2022}, \cite{wangRoad6GVisions2023}. ISAC could not only share the same wireless infrastructures to enhance the hardware utilization efficiency but also provide sensing capabilities for wireless networks to obtain information about the surrounding environment \cite{sunModelingQuantitativeAnalysis2022, huangJointLocalizationEnvironment2023, tongEnvironmentSensingConsidering2022}. Dual-function design of communication and radar \cite{liuJointRadarCommunication2020} is a hot issue in ISAC research, which transmits communication information to the receiver and extracts sensing data from the scattered echoes, such as position and velocity. With the development of emerging services such as autonomous driving, industrial internet of things, and digital twin \cite{cuiIntegratingSensingCommunications2021}, future dual-functional radar and communication (DFRC) \cite{liuJointTransmitBeamforming2020} systems are expected to provide sensing applications, e.g., target tracking \cite{zhongEmpoweringV2XNetwork2022}, intrusion detection, and environmental monitoring \cite{wuWiTrajRobustIndoor2023}. In the industry, ISAC is developed in Huawei \cite{liIntegratedSensingCommunication2021a} and ZTE \cite{maHighlyEfficientWaveform2023} for wireless network investigations. Regarding standardization, 3GPP SA1 has established research project TR 22.837 \cite{3gppStudyIntegratedSensing2022} to discuss the use cases and potential requirements for ISAC while focusing on security and privacy.

\IEEEpubidadjcol

The early development of ISAC focused on simultaneously achieving communication and radar functions, with an emphasis on radar sensing \cite{zhangEnablingJointCommunication2021}. Its evolution went through stages such as coexistence of communication and radar \cite{zhengRadarCommunicationCoexistence2019}, radar-centric communication, and communication-centric radar, gradually evolving into ISAC \cite{chiriyathRadarCommunicationsConvergenceCoexistence2017}.
In the ISAC systems, the symbiotic relationship between communication and sensing is the linchpin for unlocking their full potential. As these systems evolve from the early stages of exploring communication-radar dynamics to the sophisticated integration of communication and sensing, the need to analyze their performance in real-world scenarios becomes paramount, which is the focus of this paper.

\subsection{Related Works}
The research on ISAC performance analysis transitioned from early single-point systems performance to network-level systems performance. Initial studies concentrated on the ISAC in individual devices or systems, emphasizing how the fusion of these functions could enhance spectrum and energy efficiencies \cite{douChannelSharingAided2023, wuEnergyEfficientMIMOIntegrated2024}, reduce hardware and signal processing costs \cite{zhangOverviewSignalProcessing2021}, and explore the performance limits in information theory \cite{liuSurveyFundamentalLimits2022}.

The characterization of the trade-off between communication and sensing performance is fundamental to ISAC theory. In \cite{liOuterBoundsJoint2021}, the authors derived Cramér-Rao bounds (CRB) for direction, distance, and velocity estimation in single-station uplink scenario, establishing the trade-off relationship between estimation performance and communication rate. Similarly, \cite{ouyangPerformanceDownlinkUplink2022} analyzed the performance of downlink and uplink ISAC systems from an information-theoretic perspective, discussing communication rate (CR), sensing rate (SR), and achievable CR-SR regions in the scenario where a single base station (BS) serves multiple users. In \cite{xiongFundamentalTradeoffIntegrated2023}, the authors considered a point-to-point model in a dual-function ISAC transmitter scenario, depicting the two-fold trade-offs between achievable communication rates and achievable CRB. Some researchers dedicated efforts to integrated waveform design, where a full-duplex ISAC waveform scheme was devised, and the communication and sensing performances of a half-duplex single-base node were analyzed and optimized \cite{xiaoWaveformDesignPerformance2022}. For massive MIMO-ISAC systems, J-PoTdCe significantly enhances target detection and channel estimation performance by leveraging the pilot beamforming gain and joint burst sparsity of radar and communication channels \cite{huangJointPilotOptimization2022}. Additionally, a unified channel and target parameter estimation algorithm was proposed, which utilizes an iterative estimation method to improve sensing accuracy while enhancing channel estimation performance \cite{zhangIntegratedSensingCommunication2024}.

With the continuous advancement of wireless communication and sensing technologies, wireless networks are evolving towards greater density, leading to the extension of ISAC research to the network level \cite{chengIntegratedSensingCommunications2022, mengNetworkLevelIntegratedSensing2023}. This involves collaborative work among multiple devices, sensors, or nodes to construct more robust and intelligent communication-sensing networks \cite{weiMultiFunctional6GWireless2022}. At the network level, ISAC faces challenges due to the simultaneous communication and sensing tasks of multiple nodes, making the system more susceptible to interference. The rational analysis of inter-BS interference distribution and accurate representation of network performance are crucial for guiding network planning and optimization \cite{liuDistributedUnsupervisedLearning2023}. Stochastic geometry is commonly employed to model and analyze the spatial layout and location distribution of nodes in wireless networks, providing a powerful analytical framework \cite{hmamoucheNewTrendsStochastic2021, zhangImpactsAntennaDowntilt2023, okatiDownlinkCoverageRate2020}. In \cite{renPerformanceTradeoffsJoint2019}, the authors considered interference constraints in radar-communication networks, deriving the trade-off between radar detection range and network throughput in shared wireless channel scenarios.

However, in practical deployment, especially in urban environments, the previous works often insufficiently consider the blockage effects introduced by buildings. These blockage effects make interference analysis between nodes complex, increasing the difficulty of system performance analysis. In \cite{olsonCoverageCapacityJoint2022}, a new mathematical framework was proposed to characterize the coverage probability and ergodic capacity of communication and sensing in ISAC networks. The study analyzed the downlink sensing and communication coverage range and capacity of an ISAC network implemented with OFDM radar. In \cite{skouroumounisFDJCASTechniquesMmWave2022}, the spatial distribution of BSs was modeled as a $\beta$-Ginibre Point Process ($\beta$-GPP), providing insights into the collaborative detection performance of an FD-JCAS system in a heterogeneous millimeter-wave cellular network. In \cite{maPerformanceCooperativeDetection2023}, the distribution of blockages in urban areas was modeled by analyzing geographical information system (GIS) measurement statistics, and the successful detection probability and communication probability of ISAC systems were derived. When the impact of blockage was considered in these literatures, the exponential blockage models, generalized sphere models, average probability approaches, and similar methods are always applied to simplify the analysis.

\subsection{Our Contributions}
This paper addresses a critical research gap by examining the coverage performance of ISAC networks in the presence of blockage effects. Unlike existing studies that primarily focus on single-node analysis, our research extends to a network-level perspective, providing insights into the complex dynamics of communication-sensing interactions under realistic urban conditions. In this context, this paper will focus on analyzing the coverage performance of ISAC networks with blockage effects, proposing a general analytical model that is easy to handle. This model avoids mutual interference problems between communication and sensing through the design of time-division duplexing. By developing easy-to-implement special cases and closed-form expressions, we facilitate the practical evaluation and optimization of ISAC networks, thereby advancing the state-of-the-art in this field. The main contributions of this paper are summarized as follows:

\begin{itemize}
    \item Firstly, a generalized framework is proposed through stochastic geometry theory while considering the inherent randomness of BS deployment, signal propagation, target reflection, environmental blockages, and interference. A distance-dependent blockage model is adopted in the urban scenario, which integrates the effects of line of sight (LoS), non-line of sight (NLoS), and target reflection cascading (TRC) paths.
    \item Secondly, the proposed framework facilitates an evaluation of the impact on coverage performance in ISAC networks. This evaluation entails the resolution of power loss distribution resulting from link blockage and potential interference for each path independently. Subsequently, expressions for both the communication and sensing signal-to-interference-plus-noise ratio (SINR) and coverage probability are derived.
    \item Thirdly, special cases such as no blockage and no noise are analyzed. Based on stripping out the influencing factors, closed-form expressions that are easier to compute are derived. It is observed that in these idealized conditions, the communication coverage probability can be independent of the density of BS deployment, but the conditions that make the sensing BS-independent special case will be more demanding.
    \item Finally, the theoretical accuracy is verified by simulation. The simulation results demonstrate that increasing the BS deployment density can help improve the communication and sensing coverage probabilities when the BS deployment density is lower than the blocking density. The results also show that there exists an optimal BS density that maximizes the communication or sensing coverage probability when using the nearest visible association strategy, while blocking can provide a favorable gain in coverage probability by blocking more interference.
\end{itemize}

\subsection{Outline of the Paper}
The rest of this paper is organized as follows. Section II introduces the general ISAC system model with multiple BSs. Then, Section III evaluates the communication performance and sensing performance, in which the coverage probabilities under various scenarios are derived. Finally, Section IV provides numerical results to verify our analysis before concluding the paper in Section V. 

\textit{Notations:} We use regular, italic, bold lowercase, and bold uppercase letters to denote function names, scalars, vectors, and matrices, respectively. $\mathbb{E}(\cdot)$ denotes the expectation operation. $I_m(\cdot)$ denotes the modified Bessel function of the first kind of order $m$. $a!$ is the factorial of the nonnegative integer $a$. $\Gamma(\cdot)$ denotes the gamma function. $_2F_1(\cdot)$ denotes the Gauss hypergeometric function. $\mathrm{erfc}(\cdot)$ is the compensation error function. For quick reference, the main mathematical notations used in this paper are summarized in Table I.

\begin{table}\caption{Summary of Notations\label{tab:table1}}
    \centering
    \begin{tabular}{|m{1.9cm}<{\centering}|m{6cm}|}
        \hline
        \textbf{Notation}                     &  \multicolumn{1}{c|}{\textbf{Description}}  \\ \hline
        $\Phi_{\mathrm{bs}}$                  & Set of locations of all ISAC BSs except $b_0$. \\ \hline
        $\lambda_{\mathrm{bs}}$               & Density of ISAC BSs. \\ \hline
        $k_L$, $k_N$, $k_R$                   & Gain coefficients of the LoS path, NLoS path and echo path. \\ \hline
        $\alpha_L$, $\alpha_N$, $\alpha_R$    & Path loss exponents of the LoS path, NLoS path and echo path. \\ \hline
        $g^{\mathrm{sens}}_{L,i}$, $g^{\mathrm{sens}}_{N,i}$, $g^{\mathrm{comm}}_{L,i}$, $g^{\mathrm{comm}}_{N,i}$  & Small-scale fading gain of the LoS and NLoS interference link from the $i$-th BS to the sensing and communication associated BS. \\ \hline
        $\sigma_\mathrm{rcs}$                 & Radar cross-section of the ST. \\ \hline
        $b_i$                                 & The $i$-th ISAC BS, where $b_0$ denotes the sensing and communication associated BS. \\ \hline
        $R_i$, $\tilde{R}_i$, $\hat{R}_i$     & Distance between $b_i$ and UE / ST / $b_0$, where $R_0$ / $\tilde{R}_0$ denotes the distance between $b_0$ and UE / ST, and $\hat{R}_0$ means nothing. \\ \hline
        $M_i$, $\tilde{M}_i$, $\hat{M}_i$     & Independent Bernoulli random variables indicating whether there is a LoS path between $b_i$ and UE / ST / $b_0$, where $S_0 = \tilde{M}_0 = 1$, and $\hat{M}_0$ means nothing. \\ \hline
        $\mathrm{PL}_L$, $\mathrm{PL}_N$, $\mathrm{PL}_R$  & Large-scale path loss functions for LoS, NLoS and echo paths. \\ \hline
        $\mathrm{Pr}_L$, $\mathrm{Pr}_N$      & Probability of LoS and NLoS propagation. \\ \hline
        $\beta$                               & The parameter related to the shape and size of the blockage. \\ \hline
        $p$                                   & The proportion of the area in the region that is blocked. \\ \hline
    \end{tabular}
\end{table}
\section{System Model} \label{SystemModel}

In our proposed ISAC system, as illustrated in Fig. \ref{fig1}, a multi-BS environment coexists with user equipment (UEs) and sensing targets (STs), where ISAC-enabled BSs perform dual functionalities. The locations of BS are modeled as a Poisson point process (PPP) on the $\mathbb{R}^2$ plane with density of $\lambda_{\mathrm{bs}}$, denoted as $\Phi_{\mathrm{bs}}$.

The sensing aspect of the ISAC system leverages radar principles, capitalizing on echoes from transmitted pulses to discern and localize STs. To maintain clarity in the initial explanation, the STs are simplified as point scatterers, primarily detectable under Line-of-Sight (LoS) conditions. After transmitting a dedicated sensing signal, the BSs await the return of echoes from the STs. These echoes carry imprints of the target's position, velocity, and potentially other characteristics. Using the processed echoes, target parameters are estimated. The communication operations embrace a versatile model that considers both LoS and Non-Line-of-Sight (NLoS) propagation paths, ensuring comprehensive coverage and reliability. On this basis, the typical ST and the typical UE are both associated with the nearest visible BS. For simplicity, it is assumed that the sensing-associated BS and the communication-associated BS are the same and denoted as $b_0$.

To synchronize these intertwined operations, wireless resources are temporally divided into communication slots and sensing slots, with each BS synchronized through wired connections to ensure simultaneous execution of communication or sensing tasks within the same time slot. Assuming that each sensing echo returns within the sensing slot, which isolates mutual interference between communication and sensing. Additionally, the slot allocation ratio can be adjusted based on different service requirements to enhance the flexibility of communication sensing functions.

\begin{figure}[t]
    \centering
    \includegraphics[width=8.5cm]{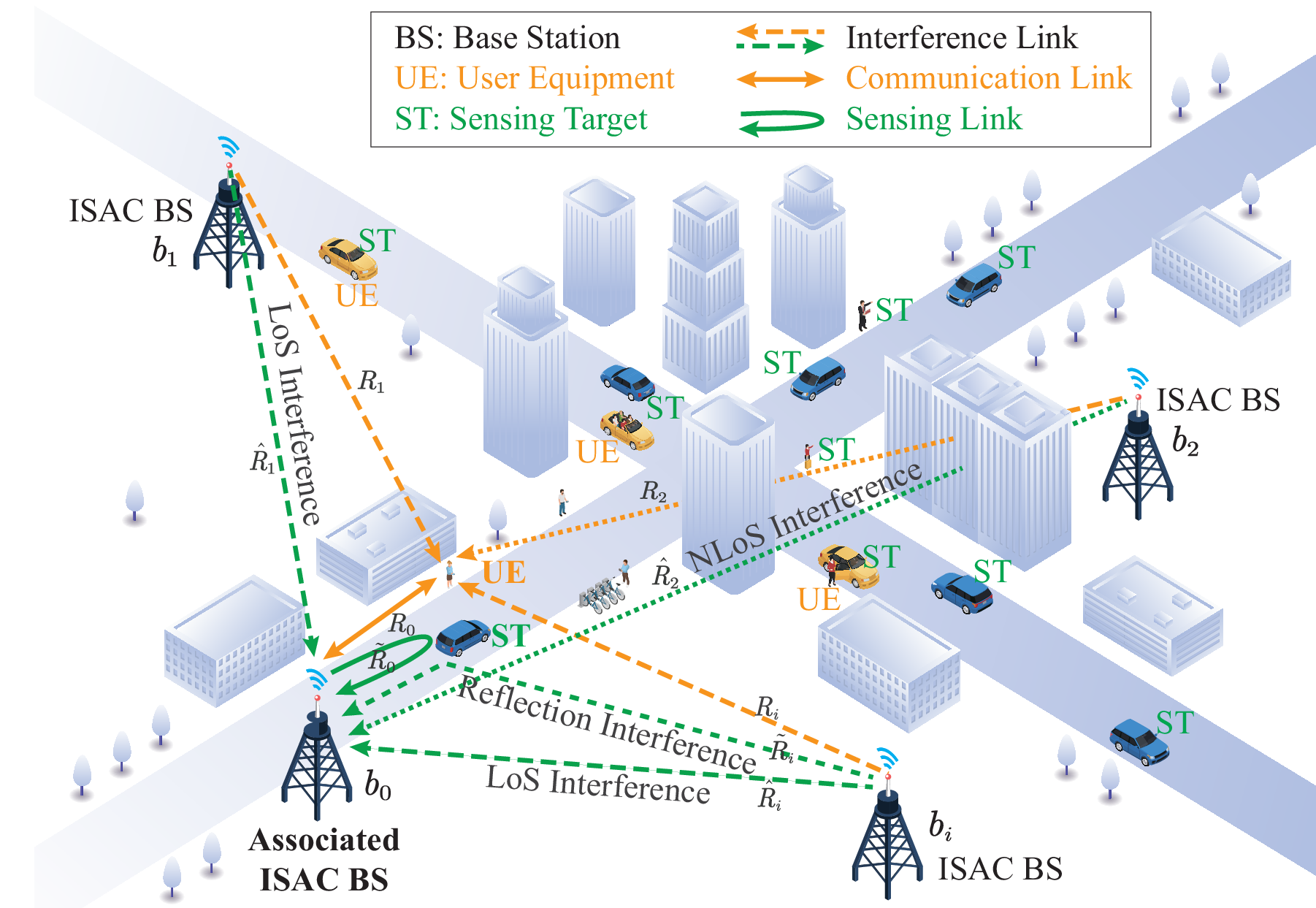}
    \caption{Schematic diagram of the considered system model.}\label{fig1}
\end{figure}

\subsection{Path Loss and Blockage Model}
The large-scale path loss functions for LoS and NLoS paths can be expressed as
\begin{equation}
    \mathrm{PL}_L(r)=k_L r^{-\alpha_L},\;\;\;\;
    \mathrm{PL}_N(r)=k_N r^{-\alpha_N},
\end{equation}
where $k_L$ and $k_N$ are coefficients of the LoS path and NLoS path, and $\alpha_L$ and $\alpha_N$ are the path loss exponents for both. Similarly, the sensed LoS echo path loss can be derived from the monostatic radar equation as
\begin{equation}
    \mathrm{PL}_R(r)=k_R r^{-\alpha_R},
\end{equation}
where $k_R = k_L/4\pi$ is the gain coefficient of the echo path, and $\alpha_R=2\alpha_L$ is its road loss exponent. 

Assume a Boolean scheme for blockage formation of rectangles. The centers of the rectangles form a homogeneous PPP with density $\lambda_\mathrm{bk}$. The length $L_k$ and the width $W_k$ of the rectangles are assumed to be i.i.d. According to Theorem I in \cite{baiAnalysisBlockageEffects2014}, the number of blockages is a Poisson random variable with the mean $\beta R + p$. The probability of being blocked is related to the link length, and the LoS and NLoS probabilities for a link of length $R$ can be modeled as
\begin{equation}
    \mathrm{Pr}_L(r)=e^{-(\beta r + p)},\;\;\;\;
    \mathrm{Pr}_N(r)=1-\mathrm{Pr}_L(r),
\end{equation}
where $\beta = 2\lambda_\mathrm{bk} (\mathbb{E}[L]+\mathbb{E}[W])/\pi$ is the parameter related to the shape and size of the blockage and $p = \lambda_\mathrm{bk} \mathbb{E}[L]\mathbb{E}[W]$ is the proportion of the area in the region that is blocked. Specific typical values can be selected for different scenarios according to the 3GPP TS 38.901 standard \cite{3gppStudyChannelModel2020}.

Assuming that the number of blockages on the link is independent, the probability density function (PDF)\footnote{Note: This function does not satisfy the normalization property due to the presence of blockages. Therefore it should not be called a PDF rigorously. However, the probability hypothesis density (PHD) can be used to describe this non-normalized measure.} of the distance $R_0$ (or $\tilde{R}_0$) from the UE (or ST) to the nearest visible BS is then \cite{baiAnalysisBlockageEffects2014}
\begin{equation}\label{eq:fR0PDF}
    f_{R_0}(r) = 2\pi \lambda_{\mathrm{bs}} r e^{-(\beta r+p +2\pi \lambda_{\mathrm{bs}} U(r))},
\end{equation}
where $U(r) = (e^{-p}/\beta^2)[1-(\beta r+1)e^{-\beta r}]$. Its integral does not satisfy the normalization condition and can be calculated as
\begin{equation}\label{eq:fR0PDFintegral}
    \begin{aligned}
    F_{R_0}(r) &= \int_{0}^{\infty}f_{R_0}(r) \mathrm{d}r 
    \\
    &= 1- \exp\left(-\frac{2e^{-p}\lambda_\mathrm{bs}\pi}{\beta^2}\right).
    \end{aligned}
\end{equation}

In particular, when both $\beta$ and $p$ tend to 0,
\begin{equation}
    U^o(r)=\lim_{p\rightarrow 0}\lim_{\beta\rightarrow 0} U(r) = \frac{r^2}{2},
\end{equation}
the nearest visible association degenerates into the nearest association policy. The PDF of the distance $R_0$ (or $\tilde{R}_0$) from the UE (or ST) to the nearest BS is then \cite{andrewsTractableApproachCoverage2011}
\begin{equation}
    \begin{aligned}
        f^o_{R_0}(r) &= \lim_{p\rightarrow 0}\lim_{\beta\rightarrow 0}f_{R_0}(r) 
        \\
        &= 2\pi \lambda_{\mathrm{bs}} r e^{-\lambda_{\mathrm{bs}} \pi r^2}.
    \end{aligned}
\end{equation}

\subsection{Desired Signal Model}

Assuming that the communication-associated BS transmits a constant mode unit signal and that there is a LoS path between the UE and $b_0$, the downlink communication received signal is

\begin{equation}
    S^{\mathrm{comm}}=g^{\mathrm{comm}}_{L,0}\mathrm{PL}_L(R_0){M}_0,
\end{equation}
where $R_0$ is the distance from $b_0$ to UE, $g^{\mathrm{comm}}_{L,0}$ is its small-scale fading gain, which obeys Rician distribution. $M_0$ is a binary indicator variable used to indicate the presence of a LoS path between $b_0$ and UE. According to the previous assumptions, communication is conducted when LoS is present, so ${M}_0=1$ holds constantly.

The sensing is carried out based on the echo signal, where the transmitted signal arrives at the receiver after round-trip large-scale fading and target reflection attenuation, and the received sensing signal can be modeled as

\begin{equation}\label{eq:Ssens}
    S^{\mathrm{sens}}=\sigma_\mathrm{rcs}\mathrm{PL}_R(\tilde{R}_0)\tilde{M}_0,
\end{equation}
where $\tilde{R}_0$ is the distance from $b_0$ to ST, and $\sigma_\mathrm{rcs}$ is the ST's radar cross-section (RCS). Similarly, $\tilde{M}_0 = 1$ holds constantly, indicating the assumption that there is always a LoS path between $b_0$ and ST.

% 干扰模型
\subsection{Interference Model}

Considering the framework of the time-division system mentioned earlier, where each BS is synchronized to ensure simultaneous execution of communication or sensing tasks within the same time slot, there is no need to consider mutual interference between communication and sensing. Therefore, interference is modeled as self-interference during the execution of their respective tasks.

The interference of the communication considers the downlink interference at the UE, so the UE at the origin is chosen as a typical UE. The interference it suffers is mainly composed of two parts, interference from LoS paths $I^{\mathrm{comm}}_\mathrm{LoS}$ and interference from NLoS paths $I^{\mathrm{comm}}_\mathrm{NLoS}$, which can be expressed as
\begin{align}
    I^{\mathrm{comm}}=I^{\mathrm{comm}}_\mathrm{LoS}+I^{\mathrm{comm}}_{\mathrm{NLoS}}           \label{eq:Icomm},
\end{align}
wherein
\begin{align}
    I^{\mathrm{comm}}_\mathrm{LoS}      & =\sum_{i\in\Phi_{\mathrm{bs}}}g^{\mathrm{comm}}_{L,i}\mathrm{PL}_L(R_i)M_i,           \tag{\ref{eq:Icomm}{a}} \label{eq:Icomm_a} \\
    I^{\mathrm{comm}}_{\mathrm{NLoS}} & =\sum_{i\in\Phi_{\mathrm{bs}}}g^{\mathrm{comm}}_{N,i}\mathrm{PL}_N(R_i)(1-M_i),      \tag{\ref{eq:Icomm}{b}} \label{eq:Icomm_b}
\end{align}
where $b_0$ is the communication-related BS. $g^{\mathrm{comm}}_{L,i}$ and $g^{\mathrm{comm}}_{N,i}$ are the small-scale fading gains on the LoS and NLoS interference links, respectively. Furthermore, $R_i$ is the distance from $b_i$ to UE, and $M_i\sim \mathrm{Bernoulli}(\mathrm{Pr}_L(R_i))$ are independent Bernoulli random variables indicating whether there is a LoS path between $b_i$ and UE.

Since the sensing results are obtained at the BS, according to Slivnyak's theorem, without loss of generality, the BS located at the origin is chosen as the typical sensing analysis node. The interference generated by other BSs can be composed of three parts. If there is a LoS path between the interfering BS and the typical BS, the interference is denoted as $I^{\mathrm{sens}}_\mathrm{LoS}$, otherwise it is expressed as $I^{\mathrm{sens}}_{\mathrm{NLoS}}$. In addition, when there is no collaboration between BSs, the interference $I^{\mathrm{sens}}_{\mathrm{TRC}}$ generated at the sensing BS by the signal transmitted by the interfering BS after reflection from the ST should also be considered, which is referred to as the target reflection cascade signal.\footnote{Note: The TRC signal here is a different term from the previous echo signal. The echo refers to the useful signal returned by the sensing signal after being reflected by the target, whereas TRC refers to the interfering signal that arrives at the sensing BS after being reflected by the target by the signals transmitted by other BSs.} Accordingly, the aggregated sensing interference can be derived as
\begin{align}
    I^{\mathrm{sens}} =  I^{\mathrm{sens}}_\mathrm{LoS}+I^{\mathrm{sens}}_{\mathrm{NLoS}}+I^{\mathrm{sens}}_{\mathrm{TRC}},   \label{eq:Isens}
\end{align}
wherein
\begin{align}
    I^{\mathrm{sens}}_\mathrm{LoS} =    & \sum_{i\in\Phi_{\mathrm{bs}}}g^{\mathrm{sens}}_{L,i}\mathrm{PL}_L(\hat{R}_i)\hat{M}_i,                                                                \tag{\ref{eq:Isens}{a}} \label{eq:Isens_a} \\
    I^{\mathrm{sens}}_{\mathrm{NLoS}} = & \sum_{i\in\Phi_{\mathrm{bs}}}g^{\mathrm{sens}}_{N,i}\mathrm{PL}_N(\hat{R}_i)(1-\hat{M}_i),                                                         \tag{\ref{eq:Isens}{b}} \label{eq:Isens_b}    \\
    I^{\mathrm{sens}}_{\mathrm{TRC}} =  & \sum_{i\in\Phi_{\mathrm{bs}}}\sigma_\mathrm{rcs}k_R\tilde{R}_i^{-\alpha_L}\tilde{M}_i \times \tilde{R}_0^{-\alpha_L} \tilde{M}_0,    \tag{\ref{eq:Isens}{c}} \label{eq:Isens_c}
\end{align}
where $b_0$ is the sensing-associated BS, which is assumed to be the same as the communication-associated BS for convenience. $g^{\mathrm{sens}}_{L,i}$ and $g^{\mathrm{sens}}_{N,i}$ are the small-scale fading gains on the LoS and NLoS interference links respectively. Besides, $\hat{R}_i$ is the distance from the $i$-th interfering BS $b_i$ to $b_0$, $\tilde{R}_i$ is the distance from $b_i$ to ST. Furthermore, $\hat{M}_i\sim \mathrm{Bernoulli}(\mathrm{Pr}_L(\hat{R}_i))$ are independent Bernoulli random variables indicating whether there is a LoS path between $b_i$ and $b_0$, and $\tilde{M}_i\sim \mathrm{Bernoulli}(\mathrm{Pr}_L(\tilde{R}_i))$ denote the existence of a LoS path between $b_i$ and ST. Notice that $\tilde{M}_0=1$ holds constantly.

\subsection{Small-scale Fading Model}
For the LoS interference link, assuming that its small-scale fading obeys a Rician distribution and the corresponding instantaneous signal power obeys a scaled non-central chi-squared distribution with two degrees of freedom, its PDF can be expressed as
\begin{equation}\label{eq:RicianPDF}
    \begin{aligned}
        f_{g^{\mathrm{c/s}}_{L,i}}\left( x \right) & =\left( 1+K \right) e^{-K-\left( 1+K \right) x}I_0\left( 2\sqrt{K\left( 1+K \right)x} \right) , \\
    \end{aligned}
\end{equation}
where ``$\mathrm{c/s}$'' means ``$\mathrm{comm}$'' or ``$\mathrm{sens}$'', $K$ is the Rician factor, denoting the ratio of the power in the dominant component to the average power in the diffuse components, $I_0(\cdot)$ denotes the modified Bessel function of the first kind and zero order. The PDF can be approximated by a finite exponential series as \cite{yanGameTheoryApproach2019}
\begin{equation}\label{eq:RicianPDFApprox}
    \begin{aligned}
        f_{g^{\mathrm{c/s}}_{L,i}}\left( x \right) & \approx \sum_{n=1}^{N} w_n^Ku_n^Ke^{-u_n^Kx}\,\,(x \in [0,W]),
    \end{aligned}
\end{equation}
where $N$ is the number of terms in the exponential series to achieve the desired accuracy, usually a good approximation can be achieved by setting $N=4$. Both $w^K_n$ and $u^K_n$ are the real coefficients of the $n$-th term related to the Rician factor $K$ under the $u^K_n>0, (n=1,2,\cdots,N)$ and $\sum_{n=1}^{N}w_n^K=1$ constraints, and $[0, W]$ is the range of $x$. So, the complementary cumulative distribution function (CCDF) of the desired signal power is approximated by \cite{yangCoverageProbabilityAnalysis2015}
\begin{equation}\label{eq:RicianCCDF}
    \begin{aligned}
        \bar{F}_{g^{\mathrm{c/s}}_{L,i}}\left( x \right) & = \int_{x}^{\infty} f_{g^{\mathrm{c/s}}_{L,i}}(y) \mathrm{d} y \\
                                                          & \approx \sum_{n=1}^{N} w^K_n e^{-u^K_nx}.
    \end{aligned}
\end{equation}

For the NLoS interference link, assuming that its small-scale fading obeys a Rayleigh distribution, the corresponding instantaneous signal power obeys an exponential distribution, and its PDF and CCDF can be expressed as
\begin{equation}\label{eq:RayleighPDF}
    \begin{aligned}
        f_{g^{\mathrm{c/s}}_{N,i}}\left( x \right) & =\mu^{\mathrm{c/s}}_{N}e^{-\mu^{\mathrm{c/s}}_{N}x}
    \end{aligned}
\end{equation}
and
\begin{equation}\label{eq:RayleighCCDF}
    \begin{aligned}
        \bar{F}_{g^{\mathrm{c/s}}_{N,i}}\left( x \right) & = \int_{x}^{\infty} f_{g^{\mathrm{c/s}}_{N,i}}(y) \mathrm{d} y \\
        &=e^{-\mu^{\mathrm{c/s}}_{N}x}
    \end{aligned}
\end{equation}
respectively, where $\mu^{\mathrm{c/s}}_{N}$ is the Rayleigh fading parameter.

For self-receiving and $\mathrm{BS}_1$-sending $\mathrm{BS}_2$-receiving reflection paths, the returned signal power per pulse is assumed to be constant for the time on target during a single scan, but fluctuate independently from scan to scan, following i.i.d. exponential distributions. This corresponds to the Swerling I model commonly used in radar settings \cite{swerlingProbabilityDetectionFluctuating1960}. Under this assumption, the PDF for $\sigma_\mathrm{rcs}$ in Eqs. \eqref{eq:Ssens} and \eqref{eq:Isens_c} can be expressed as follows
\begin{equation}\label{eq:rcsPDF}
    f_{\sigma _{\mathrm{rcs}}}\left( \sigma \right) =\frac{1}{\bar{\sigma}_{\mathrm{rcs}}}e^{-\frac{\sigma}{\bar{\sigma}_{\mathrm{rcs}}}},\left( \sigma \ge 0 \right) ,
\end{equation}
where $\bar{\sigma}_{\mathrm{rcs}}$ is the average of RCS. Its CCDF can be expressed as
\begin{equation}\label{eq:rcsCCDF}
    \bar{F}_{\sigma _{\mathrm{rcs}}}( \sigma ) =\int_{\sigma}^{\infty} f_{\sigma _{\mathrm{rcs}}}(y) \mathrm{d} y=e^{-\frac{\sigma}{\bar{\sigma}_{\mathrm{rcs}}}},( \sigma \ge 0 ) .
\end{equation}

\subsection{Performance Metrics}
Communication performance is measured using the traditional coverage probability metric, which is defined as the probability that the SINR of the communication $\mathrm{SINR}^{\mathrm{comm}}$ is greater than the threshold $T^{\mathrm{comm}}$. It is related to the detection threshold, the density of BSs, the road loss exponents, and the blockage-related parameters, which can be denoted as
\begin{equation}\label{eq:pccommDef}
    \begin{aligned}
        p_{\mathrm{c}}^{\mathrm{comm}}(T^{\mathrm{ comm}},\lambda_{\mathrm{bs}},\alpha_L,\alpha_N,\beta ,p) \\
        = \mathbb{P}[\mathrm{SINR}^{\mathrm{comm}}>T^{\mathrm{comm}}].
    \end{aligned}
\end{equation}

The $\mathrm{SINR}^{\mathrm{comm}}$ can be directly expressed as
\begin{equation}
    \mathrm{SINR}^{\mathrm{comm}}=\frac{S^{\mathrm{comm}}}{I^{\mathrm{comm}}+N^{\mathrm{comm}}},
\end{equation}
where $N^{\mathrm{comm}}$ is the additive noise power of communication.

The sensing task is divided into target detection and parameter estimation, with detection performance characterized by the probability of detection and estimation performance described by the accuracy of the parameters. When the false alarm probability is fixed, the detection probability obtained using the Neyman-Pearson detection criterion is directly related to the SINR of sensing $\mathrm{SINR}^{\mathrm{sens}}$ \cite{sunModelingQuantitativeAnalysis2022}. 

The probability that $\mathrm{SINR}^{\mathrm{sens}}$ is greater than the signal processing threshold $T^{\mathrm{sens}}$ is defined as the sensing coverage probability (or successful detection probability), which can be given by
\begin{equation}\label{eq:pcsensDef}
    \begin{aligned}
        p_{\mathrm{c}}^{\mathrm{sens}}(T^{\mathrm{ sens}},\lambda_{\mathrm{bs}},\alpha_L,\alpha_N,\alpha_R,\beta ,p)
        \\
        = \mathbb{P}[\mathrm{SINR}^{\mathrm{sens}}>T^{\mathrm{sens}}].
    \end{aligned}
\end{equation}

The sensing SINR is defined as
\begin{equation}
    \mathrm{SINR}^{\mathrm{sens}}=\frac{S^{\mathrm{sens}}}{I^{\mathrm{sens}}+N^{\mathrm{sens}}},
\end{equation}
where $N^{\mathrm{sens}}$ is the additive noise power of sensing. 

In the time-division multiplexing system, assuming that the communication and sensing utilize the entire spectrum and power resources in each time slot, allowing for the evaluation of the snapshot performance of the ISAC networks at a specific moment. In this configuration, the proportion of communication to sensing time slots does not directly impact the SINR values for each measurement but may result in a macroscopic loss of achievable communication rates. Its impact on sensing detection probability is reflected in coherent or noncoherent cumulative gains, while its effect on estimation accuracy manifests in the allocation ratio of time-frequency resources and the configuration of frame structures.

\section{Performance Analysis}

\subsection{Communication Coverage Probability}
\begin{theorem}
    For specific detection threshold $T_\mathrm{comm}$, LoS and NLoS path loss exponents $\alpha_L$ and $\alpha_N$, and blockage parameters $p$ and $\beta$, the communication coverage probability, $p_{\mathrm{c}}^{\mathrm{comm}}(T^{\mathrm{ comm}},\lambda_{\mathrm{bs}},\alpha_L,\alpha_N,\beta ,p)$, can be calculated as \eqref{eq:PCcomm} at the top of the next page,
    \begin{figure*}[!htb]
        \begin{equation}\label{eq:PCcomm}
            \begin{aligned}
                p_{\mathrm{c}}^{\mathrm{comm}}(&T^{\mathrm{comm}},\lambda_{\mathrm{bs}},\alpha_L,\alpha_N,\beta ,p)=\int_0^{\infty}\sum_{n=1}^N{w_{n}^{K}}\exp \Bigg[ -\frac{u_{n}^{K}r^{\alpha _L}T^{\mathrm{comm}}N^{\mathrm{comm}}}{k_L}
                \\
                &-2\pi \lambda_{\mathrm{bs}} \Bigg[ \sum_{m=1}^N{w_{m}^{K}\mathbb{F}\left( \frac{u_{m}^{K}}{u_{n}^{K}r^{\alpha _L}T^{\mathrm{comm}}},\alpha _L,\mathrm{Pr}_L(x),r \right)}+\mathbb{F}\left( \frac{\mu_{N}^{\mathrm{comm}}k_L}{u_{n}^{K}r^{\alpha _L}T^{\mathrm{comm}}k_N},\alpha _N,\mathrm{Pr}_N(x),0 \right) \Bigg] \Bigg] f_{R_0}\left( r \right) \mathrm{d}r
            \end{aligned},
        \end{equation}
        \begin{equation}\label{eq:PCsens}
            \begin{aligned}
                p_{\mathrm{c}}^{\mathrm{sens}}(T^{\mathrm{sens}},\lambda_{\mathrm{bs}},\alpha_L,\alpha_N,\alpha_R,\beta ,p)&=\int_0^{\infty}\exp \Bigg[ -\frac{r^{\alpha _R}T^{\mathrm{sens}}N^{\mathrm{sens}}}{\bar{\sigma}_{\mathrm{rcs}}k_R}-2\pi \lambda_{\mathrm{bs}} \Bigg[ \sum_{n=1}^N{w_{n}^{K}\mathbb{F}\left( \frac{u_{n}^{K}\bar{\sigma}_{\mathrm{rcs}}k_R}{r^{\alpha _R}T^{\mathrm{sens}}k_L},\alpha _L,\mathrm{Pr}_L(x),r \right)}
                \\
                &+\mathbb{F}\left( \frac{\mu_{N}^{\mathrm{sens}}\bar{\sigma}_{\mathrm{rcs}}k_R}{r^{\alpha _R}T^{\mathrm{sens}}k_N},\alpha _N,\mathrm{Pr}_N(x),0 \right) +\mathbb{F}\left( \frac{r^{\alpha _L}}{r^{\alpha _R}T^{\mathrm{sens}}},\alpha _L,\mathrm{Pr}_L(x),r \right) \Bigg] \Bigg] f_{\tilde{R}_0}(r)\mathrm{d}r
            \end{aligned},
        \end{equation}
        \hrule
    \end{figure*}
    where
    \begin{equation}
        \mathbb{F}\left( \varepsilon ,\alpha ,p(x),h \right) =\int_h^{\infty}{\frac{x p(x)}{\varepsilon x^{\alpha}+1}\mathrm{d}x}.
    \end{equation}
\end{theorem}

\begin{proof}
    The proof is provided in Appendix A.
\end{proof}

This formula can be divided into three segments: the first addressing attenuation influenced by noise, the second focusing on attenuation affected by LoS interference, and the third concerning attenuation attributable to NLoS interference. Each segment encompasses fundamental signal attenuation over distance, the impacts of blockage, and the deployment density of BSs. This structured representation facilitates the identification and isolation of individual factors, providing a lucid framework for subsequent theoretical simplification and practical application. It not only enables the separate examination of how each source of attenuation specifically impacts system performance but also paves the way for model simplification under certain scenarios—such as in the absence of blockages or low-noise environments—thereby accelerating insights into and realization of system optimization designs.

It is interesting to note that, when $\alpha > 2$, there exists a quasi-closed-form solution to the $\mathbb{F}(\cdot)$ function 
\begin{equation} \label{eq:FSpecialCase3}
    \begin{aligned}
        \mathbb{F}\left( \varepsilon ,\alpha ,1,h \right) &=\int_h^{\infty}{\frac{x }{\varepsilon x^{\alpha}+1}\mathrm{d}x}
        \\
        &=\frac{h^{-(\alpha - 2)}{_{2}F_1}\left(1,\frac{\alpha - 2}{\alpha},2-\frac{2}{\alpha},-\frac{h^{-\alpha}}{\varepsilon }\right)}{\varepsilon (\alpha - 2)},
    \end{aligned}
\end{equation}
where ${_{2}F_1}(\cdot)$ is a hypergeometric function which can be calculated using the Gamma functions and finite integral as
\begin{equation} \label{eq:FSpecialCase2}
    \begin{aligned}
        {_{2}F_1}(a,b;c;t)=\frac{\Gamma(c)}{\Gamma(b)\Gamma(c-b)}\int_0^1\frac{z^{b-1}(1-z)^{c-b-1}}{(1-tz)^a}\mathrm{d}z
    \end{aligned},
\end{equation}
where $\Gamma(z)=\int_0^{\infty}t^{z-1}e^{-t}\mathrm{d}t$, and it can be computed by modern computers almost as quickly as the basic function.

More particularly, when $\alpha = 4$, there exists a closed-form solution to the $\mathbb{F}(\cdot)$ function as
\begin{equation} \label{eq:FSpecialCase11}
    \begin{aligned}
        \mathbb{F}\left( \varepsilon , 4, 1, h \right) &=\int_h^{\infty}{\frac{x }{\varepsilon x^{4}+1}\mathrm{d}x}
        \\
        &=\frac{\pi-2\arctan(h^2\sqrt{\varepsilon})}{4\sqrt{\varepsilon} }.
    \end{aligned}
\end{equation}

Due to the presence of blockages, the distribution of distances between the accessing BS and the served UE is related to the BS density, so even if noise is not taken into account, the derivative function of Theorem 1 with respect to $\lambda_{\mathrm{bs}}$ is still not constantly 0, and the communication coverage probability varies with the BS density $\lambda_{\mathrm{bs}}$. However, because of the existence of quasi-closed-form expressions at $\alpha>2$, the optimal BS deployment density for a given parameter can be computed quite easily.

In practical urban environments, although blockage effects are pervasive, they can be disregarded in special scenarios such as open areas or near elevated BSs, and in theoretical scenarios where BS density significantly surpasses blockage density. Under these conditions, with blockage effects minimized, system performance approaches its theoretical optimum. Analyzing communication and sensing coverage probabilities without considering blockages helps establish a performance baseline under ideal or extreme conditions. This is crucial for assessing the optimization potential of real-world network layouts and configurations.

Hence, the special case of no blockage is discussed. When $\beta\rightarrow 0$ and $p \rightarrow 0$, the probabilities of LoS and NLoS propagation become $\mathrm{Pr}_L(x)\rightarrow 1$ and  $\mathrm{Pr}_N(x)\rightarrow 0$ respectively, and we have $\mathbb{F}\left( \varepsilon ,\alpha ,0,h \right)=0$. The PDF of the distance $R_0$ from the UE to the nearest BS is then $f_{R_0}(r)\rightarrow f^o_{R_0}(r)$.

\textit{Corollary 1:} When blockage is not considered, the communication coverage probability in Theorem 1 degenerates to
\begin{equation} \label{eq:Corollary1}
    \begin{aligned}
        &p_{\mathrm{c}}^{\mathrm{comm}}(T^{\mathrm{comm}},\lambda _b,\alpha _L,\cdot,0,0)
        \\
        &={\sum_{n=1}^N{w_{n}^{K}}}\int_0^{\infty}\exp \Bigg[ -\frac{u_{n}^{K}r^{\alpha _L}T^{\mathrm{comm}}N^{\mathrm{comm}}}{k_L}
        \\
        &-2\pi \lambda _b\left[ \sum_{m=1}^N{w_{m}^{K}\mathbb{F} \left( \frac{u_{m}^{K}}{u_{n}^{K}r^{\alpha _L}T^{\mathrm{comm}}},\alpha _L,1,r \right)} \right] \Bigg] f^o_{R_0}\left( r \right) \mathrm{d}r.
    \end{aligned}
\end{equation}

In the non-blockage scenario, Corollary 1 leads to some exciting conclusions at $\alpha_L>2$ or even $\alpha_L = 4$. These two path loss exponents are selected for their representativeness. For instance, in downtown or densely built-up areas, signal propagation encounters multipath effects and building blockages, frequently resulting in a path loss exponent exceeding 2, occasionally reaching as high as 3 to 4. The path loss exponent $\alpha_L = 4$ signifies a severe attenuation scenario, pertinent to high-frequency communication systems affected by atmospheric absorption and rain fading, ideal for rapid estimation in long-range transmissions or scenarios emphasizing signal degradation.

Moreover, to concentrate on the impact of interference on system performance, simplified analysis models are explored by considering noiseless scenarios. In high signal-to-noise ratio (SNR) situations, such as deep-space communications or systems employing advanced noise reduction techniques, noise becomes negligible compared to interference. This analysis elucidates the key factors for enhancing system performance under interference-limited conditions, furnishing a theoretical foundation for designing efficient anti-interference mechanisms.

So the main simplifications of the various combinations considered are (1) noisy networks with path loss exponent $\alpha_L = 4$, (2) interference-limited no noise networks with path loss exponent $\alpha_L > 2$, and (3) interference-limited no noise networks with path loss exponent $\alpha_L = 4$.

\textbf{Special Case 1:} \textit{No blockage, Noise, $\alpha_L=4$.}
Substituting Eq. \eqref{eq:FSpecialCase11} into Eq. \eqref{eq:Corollary1} yields the coverage probability of communication in the non-blockage case as
\begin{equation} \label{eq:PCcommSpecialCase11}
    \begin{aligned}
        &p_{\mathrm{c}}^{\mathrm{comm}}(T^{\mathrm{comm}},\lambda_{\mathrm{bs}},4,\cdot,0 ,0)
        \\
        &=\int_0^{\infty}{\sum_{n=1}^N{w_{n}^{K}}}\exp \Bigg[ -\frac{u_{n}^{K}r^{4}T^{\mathrm{comm}}N^{\mathrm{comm}}}{k_L}-2\pi \lambda _br^2
        \\
        &\times\bigg[ \sum_{m=1}^N{w_{m}^{K}\frac{\pi-2\arctan(\sqrt{\frac{u_{m}^{K}}{u_{n}^{K}T^{\mathrm{comm}}}})}{4\sqrt{\frac{u_{m}^{K}}{u_{n}^{K}T^{\mathrm{comm}}}} }} \bigg] \Bigg] f_{R_0}^{o}\left( r \right) \mathrm{d}r
        \\
        &\overset{\left( a \right)}{=}\pi \lambda _b\sum_{n=1}^N{w_{n}^{K}}\int_0^{\infty}{\exp \left[ -\pi\lambda _b\left(  1+2 \theta^K_n \right) v-\vartheta^K_nv^2 \right]}\mathrm{d}v
        \\
        &\overset{\left( b \right)}{=}\frac{1}{2}\pi \lambda _b\sum_{n=1}^N{w_{n}^{K}\sqrt{\frac{\pi}{\vartheta _{n}^{K}}}\exp \left[ \left( \lambda _b\varPsi _{n}^{K} \right) ^2 \right]}\mathrm{erfc}\left[ \lambda _b\varPsi _{n}^{K} \right] ,
    \end{aligned}
\end{equation}
where (a) uses the substitution $r^2\rightarrow v$, and
\begin{equation} \label{eq:PCcommSpecialCase12}
    \begin{aligned}
        \begin{cases}
            \theta _{n}^{K}=\sum_{m=1}^N{w_{m}^{K}\frac{\pi -2\mathrm{arc}\tan\mathrm{(}\sqrt{\frac{u_{m}^{K}}{u_{n}^{K}T^{\mathrm{comm}}}})}{4\sqrt{\frac{u_{m}^{K}}{u_{n}^{K}T^{\mathrm{comm}}}}}},\\
            \vartheta _{n}^{K}=\frac{u_{n}^{K}T^{\mathrm{comm}}N^{\mathrm{comm}}}{k_L},\\
            \varPsi _{n}^{K}=\frac{\pi \left( 1+2\theta _{n}^{K} \right)}{2\sqrt{\vartheta _{n}^{K}}},\\
        \end{cases}
    \end{aligned}
\end{equation}
and (b) is according to 
\begin{equation} \label{eq:PCcommSpecialCase13}
    \begin{aligned}
            \int_0^{\infty}{e^{-x-bx^2}}\mathrm{d}x=\frac{1}{2}\sqrt{\frac{\pi}{b}}\exp \left[ \left( \frac{a}{2\sqrt{b}} \right) ^2 \right] \mathrm{erfc}\left[ \frac{a}{2\sqrt{b}} \right] ,
    \end{aligned}
\end{equation}
wherein $\mathrm{erfc}(x)=\frac{2}{\sqrt{\pi}}\int_x^{+\infty}e^{-t^2}\mathrm{d}t$ is the compensation error function.

This expression is effortless to compute, requiring only a simple $\mathrm{erfc}(x)$ value, and can actually be described as a closed-form expression. This case illustrates that even under high path loss conditions with the presence of noise, favorable communication coverage probabilities can be maintained through the strategic tuning of system parameters. It is well-suited for rapid estimation in scenarios involving long-distance propagation or those that emphasize signal attenuation effects.

Further, some interesting conclusions can be obtained under this model when noise is not taken into account.

\textbf{Special Case 2:} \textit{No blockage, No Noise, $\alpha_L>2$.}
When the blockage is not considered and the noise power $N^{\mathrm{comm}} \rightarrow 0$ or the interference is much larger than the noise, and the path loss exponent $\alpha_L$ is greater than 2, the communication coverage probability is
\begin{equation} \label{eq:PCcommSpecialCase21}
    \begin{aligned}
        &p_{\mathrm{c,noiseless}}^{\mathrm{comm}}(T^{\mathrm{comm}},\lambda_{\mathrm{bs}},\alpha_L,\cdot,0 ,0)
        \\
        &=\int_0^{\infty}{\sum_{n=1}^N{w_{n}^{K}}}\exp \Bigg[ -2\pi \lambda _b
        \\
        &\times\left[ \sum_{m=1}^N{w_{m}^{K}\mathbb{F} \left( \frac{u_{m}^{K}}{u_{n}^{K}r^{\alpha _L}T^{\mathrm{comm}}},\alpha _L,1,r \right)} \right] \Bigg] f^o_{R_0}\left( r \right) \mathrm{d}r
        \\
        &\overset{\left( c \right)}{=}\pi \lambda _b\sum_{n=1}^N{w_{n}^{K}}\int_0^{\infty}\exp \left[ -\pi \lambda _b v\right] {\exp \left[ -2\pi \lambda _b \xi _{n}^{K} v \right]}\mathrm{d}v
        \\
        &=\sum_{n=1}^N{w_{n}^{K}\frac{1}{1+2 \xi _{n}^{K}}},
    \end{aligned}
\end{equation}
where (c) uses the substitution $r^2\rightarrow v$, and
\begin{equation} \label{eq:PCcommSpecialCase22}
    \begin{aligned}
        \xi _{n}^{K} = \sum_{m=1}^{N}w_{m}^{K}\frac{_2F_1\left( 1,\frac{\alpha _L-2}{\alpha _L},2-\frac{2}{\alpha _L},-\frac{u_{n}^{K}T^{\mathrm{comm}}}{u_{m}^{K}} \right) }{\frac{u_{m}^{K}}{u_{n}^{K}T^{\mathrm{comm}}}(\alpha _L-2)}.
    \end{aligned}
\end{equation}

In this case, the communication coverage probability is independent of the BS density $\lambda_{\mathrm{bs}}$. It only relates to the detection threshold $T^{\mathrm{comm}}$ and the path loss exponent $\alpha_L$. This underscores that under high SNR conditions, enhancements in system performance predominantly stem from optimizing the detection threshold and accounting for signal propagation characteristics, rather than merely augmenting the number of BSs. It serves as a fitting approach for swiftly estimating the upper bounds of performance under idealized conditions.

\textbf{Special Case 3:} \textit{No blockage, No Noise, $\alpha_L=4$.}
 Combining the first two special cases and substituting $\alpha_L=4$ make the closed-form expression as follows
\begin{equation} \label{eq:PCcommSpecialCase31}
    \begin{aligned}
        &p_{\mathrm{c,noiseless}}^{\mathrm{comm}}(T^{\mathrm{comm}},\lambda_{\mathrm{bs}},4,\cdot,0 ,0)
        \\
        &=\int_0^{\infty}{\sum_{n=1}^N{w_{n}^{K}}}\exp \Bigg[ -2\pi \lambda _b r^2
        \\
        &\times\bigg[ \sum_{m=1}^N{w_{m}^{K}\frac{\pi-2\arctan(\sqrt{\frac{u_{m}^{K}}{u_{n}^{K}T^{\mathrm{comm}}}})}{4\sqrt{\frac{u_{m}^{K}}{u_{n}^{K}T^{\mathrm{comm}}}} }} \bigg] \Bigg] f_{R_0}^{o}\left( r \right) \mathrm{d}r
        \\
        &=\pi \lambda _b\sum_{n=1}^N{w_{n}^{K}}\int_0^{\infty}{\exp \left[ -\pi\lambda_{\mathrm{bs}}\left(  1+2 \theta^K_n \right) v \right]}\mathrm{d}v
        \\
        &=\sum_{n=1}^N{w_{n}^{K}\frac{1}{1+2 \theta _{n}^{K}}},
    \end{aligned}
\end{equation}
where $\theta^K_n$ is given by Eq. \eqref{eq:PCcommSpecialCase12}. 

Similar to conventional communication results, this is a very easy-to-compute closed-form expression that relates only to the detection threshold $T^{\mathrm{comm}}$. Under these highly controlled conditions—absence of both noise and blockages, alongside a severe path loss scenario typical of high-frequency transmissions—the coverage probability becomes particularly sensitive to the detection threshold $T^{\mathrm{comm}}$ and the specifics of the signal decay rate governed by $\alpha_L$, while completely disregarding the influence of BS deployment density. 

This finding reinforces the notion that in extremely challenging propagation environments, optimizing system parameters, especially the sensitivity of detection and leveraging knowledge of the propagation model, is paramount to achieving optimal performance. It presents a valuable tool for promptly assessing the limits of achievable communication reliability in demanding long-range or signal-degraded applications.

\subsection{Sensing Coverage Probability}
\begin{theorem}
    For specific detection threshold $T_\mathrm{sens}$, LoS, NLoS and echo path loss exponents $\alpha_L$, $\alpha_N$ and $\alpha_R$, and blockage parameters $p$ and $\beta$, the sensing coverage probability, $p_{\mathrm{c}}^{\mathrm{sens}}(T^{\mathrm{ sens}},\lambda_{\mathrm{bs}},\alpha_L,\alpha_N,\alpha_R,\beta ,p)$, can be calculated as \eqref{eq:PCsens} at the top of the last page.
\end{theorem}

\begin{proof}
    The proof is provided in Appendix B.
\end{proof}

Theorem 2 introduces an expanded analytical dimension with its sensing coverage probability formula, which is segmented into four crucial components. This formula not only incorporates the three attenuation aspects discussed in Theorem 1 that are pertinent to communication scenarios but also integrates a pivotal consideration specific to sensing functionality — the interference effect arising from target reflection cascading (TRC) paths. The inclusion of this additional aspect highlights the unique challenges and mechanisms inherent to sensing systems.

Similarly to the communication coverage probability, the derivative of Theorem 2 with respect to $\lambda_{\mathrm{bs}}$ is not constantly zero. So, under specific parameter configurations, an optimal deployment density of BSs that maximizes the probability of sensing coverage exists here. In open and sparse networks or in dense networks where blockage effects are neglected, the following Corollary 2 can be derived.

\textit{Corollary 2:} The sensing coverage probability in the non-blockage scenario is given by
\begin{equation} \label{eq:Corollary2}
    \begin{aligned}
        &p_{\mathrm{c}}^{\mathrm{sens}}(T^{\mathrm{sens}},\lambda _b,\alpha _L,\cdot,\alpha_R,0,0)
        \\
        &=\int_0^{\infty}\exp \Bigg[ -\frac{r^{\alpha _R}T^{\mathrm{sens}}N^{\mathrm{sens}}}{\bar{\sigma}_{\mathrm{rcs}}k_R}
        \\
        &-2\pi \lambda_{\mathrm{bs}} \Bigg[ \sum_{n=1}^N{w_{n}^{K}\mathbb{F}\left( \frac{u_{n}^{K}\bar{\sigma}_{\mathrm{rcs}}k_R}{r^{\alpha _R}T^{\mathrm{sens}}k_L},\alpha _L,1,r \right)}
        \\
        &+\mathbb{F}\left( \frac{r^{\alpha _L}}{r^{\alpha _R}T^{\mathrm{sens}}},\alpha _L,1,r \right) \Bigg] \Bigg] f^o_{\tilde{R}_0}(r)\mathrm{d}r.
    \end{aligned}
\end{equation}

When discussing whether Corollary 2 can yield closed-form solutions, we find that it is challenging to directly obtain such solutions based solely on the current configuration parameters. Therefore, additional simplifications are necessary. Firstly, we can ignore the interference effects caused by the reflected paths. This simplification is feasible in specific contexts, such as within the framework of cooperative sensing networks, where signals transmitted by other BSs can be beneficially integrated even after being reflected by targets, rather than being considered as sources of interference. Secondly, we assume that the path loss exponents for direct and reflected paths are equal. Although this assumption may seem unrealistic intuitively because reflected waves typically travel a longer distance (usually twice the direct wave's distance) and naturally suffer more attenuation without compensation mechanisms. However, such an assumption holds significant value for delving deeper into how different path interferences affect the coverage range and mechanisms of sensing systems.

Drawing from the analysis method used for communication special cases, we similarly construct several extreme scenarios for sensing, aiming to reveal the underlying principles and optimization approaches of sensing performance under more stringent environmental constraints through these idealized settings. This analytical strategy not only contributes to theoretical advancements but also provides theoretical support for the development of efficient sensing technologies to address complex environmental challenges.

\textbf{Special Case 4:} \textit{ No blockage, Noise, $I^{\mathrm{sens}}_{\mathrm{TRC}}=0$, $\alpha_L=\alpha_R=4$.}
Substituting Eq. \eqref{eq:FSpecialCase11} into Eq. \eqref{eq:Corollary2}, and assuming $\mathbb{F}\left( \frac{r^{\alpha _L}}{r^{\alpha _R}T^{\mathrm{sens}}},\alpha _L,1,r \right)=0$ yields the coverage probability of sensing in this case as
\begin{equation} \label{eq:FSpecialCase12}
    \begin{aligned}
        &p_{\mathrm{c}}^{\mathrm{sens}}(T^{\mathrm{sens}},\lambda_{\mathrm{bs}},4,\cdot,4,0 ,0)
        \\
        &=\int_0^{\infty}\exp \Bigg[ -\frac{r^{4}T^{\mathrm{sens}}N^{\mathrm{sens}}}{\bar{\sigma}_{\mathrm{rcs}}k_R}-2\pi \lambda _br^2
        \\
        &\times\Bigg[ \sum_{n=1}^N w_{n}^{K}\frac{\pi -2\mathrm{arc}\tan\left(\sqrt{\frac{u_{n}^{K}\bar{\sigma}_{\mathrm{rcs}}k_R}{T^{\mathrm{sens}}k_L}}\right)}{4\sqrt{\frac{u_{n}^{K}\bar{\sigma}_{\mathrm{rcs}}k_R}{T^{\mathrm{sens}}k_L}}} \Bigg] \Bigg] f_{\tilde{R}_0}^{o}\left( r \right) \mathrm{d}r
        \\
        &\overset{\left( d \right)}{=} \pi \lambda _b\int_0^{\infty}{\exp}\left[ -\pi \lambda _b (1+2\theta)  v-\vartheta v^2 \right] \mathrm{d}v
        \\
        &\overset{\left( e \right)}{=}\frac{1}{2}\pi \lambda _b\sqrt{\frac{\pi}{\vartheta}}\exp \left[ \left( \lambda _b\varPsi \right) ^2 \right]\mathrm{erfc}\left[ \lambda _b\varPsi \right] ,
    \end{aligned}
\end{equation}
where (d) uses the substitution $r^2\rightarrow v$, (e) uses Eq. \eqref{eq:PCcommSpecialCase13}, and
\begin{equation} \label{eq:FSpecialCase13}
    \begin{aligned}
        \begin{cases}
            \theta =\sum_{n=1}^N{w_{n}^{K}\frac{\pi -2\mathrm{arc}\tan \left( \sqrt{\frac{u_{n}^{K}\bar{\sigma}_{\mathrm{rcs}}k_R}{T^{\mathrm{sens}}k_L}} \right)}{4\sqrt{\frac{u_{n}^{K}\bar{\sigma}_{\mathrm{rcs}}k_R}{T^{\mathrm{sens}}k_L}}}},\\
            \vartheta =\frac{T^{\mathrm{sens}}N^{\mathrm{sens}}}{\bar{\sigma}_{\mathrm{rcs}}k_R},\\
            \varPsi =\frac{\pi \left( 1+2\theta \right)}{2\sqrt{\vartheta}}.\\
        \end{cases}
    \end{aligned}
\end{equation}

This expression is similar to the result in Special Case 1 and is even more straightforward to compute. Further, an even simpler expression exists for the calculation of Corollary 2 when the noise is not taken into account.

\textbf{Special Case 5:} \textit{No blockage, No Noise, $I^{\mathrm{sens}}_{\mathrm{TRC}}=0$, $\alpha_L=\alpha_R>2$.}
When the blockage is not considered and the noise power $N^{\mathrm{snes}} \rightarrow 0$ or the interference is much larger than the noise and the path loss exponent is greater than 2, the sensing coverage probability is
\begin{equation} \label{eq:FSpecialCase21}
    \begin{aligned}
        &p_{\mathrm{c,noiseless}}^{\mathrm{sens}}(T^{\mathrm{sens}},\lambda_{\mathrm{bs}},\alpha_L,\cdot,\alpha_L,0 ,0)
        \\
        &=\int_0^{\infty}\exp \Bigg[ -2\pi \lambda _b
        \\
        &\times\left[ \sum_{n=1}^N{w_{n}^{K}\mathbb{F} \left( \frac{u_{n}^{K}\bar{\sigma}_{\mathrm{rcs}}k_R}{r^{\alpha _L}T^{\mathrm{sens}}k_L},\alpha _L,1,r \right)} \right] \Bigg] f^o_{\tilde{R}_0}\left( r \right) \mathrm{d}r
        \\
        &\overset{\left( f \right)}{=}\pi \lambda _b\int_0^{\infty}\exp \left[ -\pi \lambda _b (1+2\xi) v\right] \mathrm{d}v
        \\
        &=\frac{1}{1+2 \xi },
    \end{aligned}
\end{equation}
where (f) uses the substitution $r^2\rightarrow v$, and
\begin{equation} \label{eq:FSpecialCase22}
    \begin{aligned}
        \xi = \sum_{n=1}^N{w_{n}^{K}\frac{_2F_1\left( 1,\frac{\alpha _L-2}{\alpha _L},2-\frac{2}{\alpha _L},-\frac{T^{\mathrm{sens}}k_L}{u_{n}^{K}\bar{\sigma}_{\mathrm{rcs}}k_R} \right)}{\frac{u_{n}^{K}\bar{\sigma}_{\mathrm{rcs}}k_R}{T^{\mathrm{sens}}k_L}(\alpha _L-2)}}.
    \end{aligned}
\end{equation}

In this case, the sensing coverage probability is independent of the BS density $\lambda_{\mathrm{bs}}$ and is only related to the detection threshold $T^{\mathrm{sens}}$ and the path loss exponent $\alpha_L$.

\textbf{Special Case 6:} \textit{No blockage, No noise, $I^{\mathrm{sens}}_{\mathrm{TRC}}=0$, $\alpha_L=\alpha_R=4$.}
 Combining special cases 4 and 5 leads to a closed-form solution expressed as
\begin{equation} \label{eq:FSpecialCase31}
    \begin{aligned}
        &p_{\mathrm{c}}^{\mathrm{sens}}(T^{\mathrm{sens}},\lambda_{\mathrm{bs}},4,\cdot,4,0 ,0)
        \\
        &=\int_0^{\infty}\exp \Bigg[-2\pi \lambda _br^2
        \\
        &\times\Bigg[ \sum_{n=1}^N w_{n}^{K}\frac{\pi -2\mathrm{arc}\tan\left(\sqrt{\frac{u_{n}^{K}\bar{\sigma}_{\mathrm{rcs}}k_R}{T^{\mathrm{sens}}k_L}}\right)}{4\sqrt{\frac{u_{n}^{K}\bar{\sigma}_{\mathrm{rcs}}k_R}{T^{\mathrm{sens}}k_L}}} \Bigg] \Bigg] f_{\tilde{R}_0}^{o}\left( r \right) \mathrm{d}r
        \\
        &= \pi \lambda _b\int_0^{\infty}{\exp}\left[ -\pi \lambda _b(1+2\theta)v \right] \mathrm{d}v
        \\
        &=\frac{1}{1+2 \theta },
    \end{aligned}
\end{equation}
where $\theta$ is given by Eq. \eqref{eq:FSpecialCase13}. Similar to conventional communication results, this is a very easy-to-compute closed-form expression that relates only to the detection threshold $T^{\mathrm{sens}}$.

\section{Simulation Results and Analysis}
In this section, numerical and Monte Carlo simulation results are given to evaluate the performance of communication and sensing. Unless otherwise specified, the generic parameters are set as follows: the area radius is set to be $1\,\mathrm{km}$ \cite{jiangCommunicationComputationAssisted2022}, the BS density is set to be $10^{-5}\,\mathrm{m}^{-2}$ \cite{ganCoverageRateAnalysis2024}, the system bandwidth is set to be $100\,\mathrm{MHz}$, the noise power spectral density is set to be $-174\,\mathrm{dBm/Hz}$ \cite{yanGameTheoryApproach2019}, the transmit power is set to be $43\,\mathrm{dBm}$ \cite{yanGameTheoryApproach2019}, the gain coefficients of LoS, NLoS and echo paths are set to $-75\,\mathrm{dB}$, $-90\,\mathrm{dB}$, and $-86\,\mathrm{dB}$, respectively \cite{olsonCoverageCapacityJoint2022}, and the path loss exponents are set to $2$, $3.2$, $4$, respectively \cite{olsonCoverageCapacityJoint2022}, the Rayleigh fading parameters are set to $1$, the average RCS of the ST is set to $20\,\mathrm{dBsm}$ (i.e. $100\,\mathrm{m}^{2}$ in radar settings \cite{skolnikIntroductionRadarSystems2002}), the values of the Ricain fading coefficients $u^K_n$ and $w^K_n$, mainly borrowed from \cite{yangCoverageProbabilityAnalysis2015}, are shown in Table \ref{tab:table2}. The main parameters are listed in Table \ref{tab:table3}.

\begin{table}\caption{Ricain fading simulation parameters} \label{tab:table2}
    \centering
    \begin{tabular}{|c|cc|cc|cc|}
        \hline
        \multirow{2}{*}{\begin{tabular}[c]{@{}c@{}}Term\\ Index\end{tabular}} & \multicolumn{2}{c|}{$K=1$}             & \multicolumn{2}{c|}{$K=5$}             & \multicolumn{2}{c|}{$K=10$}            \\ \cline{2-7} 
                                                                              & \multicolumn{1}{c|}{$w_n^K$} & $u_n^K$ & \multicolumn{1}{c|}{$w_n^K$} & $u_n^K$ & \multicolumn{1}{c|}{$w_n^K$} & $u_n^K$ \\ \hline
        $n=1$                                                                 & \multicolumn{1}{c|}{-0.8993} & 1.2475  & \multicolumn{1}{c|}{42.243}  & 2.9576  & \multicolumn{1}{c|}{177.75}  & 3.8741  \\ \hline
        $n=2$                                                                 & \multicolumn{1}{c|}{5.9324}  & 1.4298  & \multicolumn{1}{c|}{-189.99} & 3.7559  & \multicolumn{1}{c|}{-338.04} & 4.3761  \\ \hline
        $n=3$                                                                 & \multicolumn{1}{c|}{-5.4477} & 1.7436  & \multicolumn{1}{c|}{192.97}  & 4.1436  & \multicolumn{1}{c|}{297.00}  & 5.3985  \\ \hline
        $n=4$                                                                 & \multicolumn{1}{c|}{1.4145}  & 2.0326  & \multicolumn{1}{c|}{-44.229} & 4.7715  & \multicolumn{1}{c|}{-135.71} & 5.9937  \\ \hline
        \end{tabular}
\end{table}

\begin{table}\caption{Parameter Setting of the Network\label{tab:table3}}
    \centering
    \begin{tabular}{c|c}
        \hline \hline
        \textbf{Parameters}                                          & \textbf{Values}             \\ \hline
        Gain coefficient of the LoS path, $k_L$                      & $-75\,\mathrm{dB}$          \\ \hline
        Gain coefficient of the NLoS path, $k_N$                     & $-90\,\mathrm{dB}$          \\ \hline
        Gain coefficient of the echo path, $k_R$                     & $-86\,\mathrm{dB}$          \\ \hline
        Path loss exponent of the LoS path, $\alpha_L$               & $2$                         \\ \hline
        Path loss exponent of the NLoS path, $\alpha_N$              & $3.2$                       \\ \hline
        Path loss exponent of the echo path, $\alpha_R$              & $4$                         \\ \hline
        Rayleigh fading parameters, $\mu_{N}^{\mathrm{comm}}$,
        $\mu_{N}^{\mathrm{sens}}$                                 & $1$                         \\ \hline
        Rician fading factor, $K$                                    & $10$                        \\ \hline
        Density of ISAC BSs, $\lambda_{\mathrm{bs}}$                 & $10^{-5}\,\mathrm{m}^{-2}$  \\ \hline
        Transmit signal power, $P_{\mathrm{t}}$                      & $43\,\mathrm{dBm}$          \\ \hline
        Noise power spectral density, $P_{\mathrm{n}}$               & $-174\,\mathrm{dBm/Hz}$     \\ \hline
        Total bandwidth of the system, $B$                           & $100\,\mathrm{MHz}$         \\ \hline
        Average RCS of the ST, $\bar{\sigma}_{\mathrm{rcs}}$         & $20\,\mathrm{dBsm}$         \\ \hline
        Parameter related to the blockage, $\beta$                   & $0.008$                     \\ \hline
        Parameter related to the blockage, $p$                       & $0.1$                       \\ \hline
        Number of Monte Carlo simulations                            & $10000$                     \\ \hline
        The radius of the area                                       & $1\,\mathrm{km}$            \\ \hline
    \end{tabular}
\end{table}

The main simulation results are given in Figs. 2-5, where the lines indicate the numerical results, and the markers are the same color as the lines, indicating the Monte Carlo simulation results. In addition, solid lines are used to represent the results related to sensing; dotted lines are used to describe the results related to communication; and the size of the parameters represented by black, red, blue, green, and purple increases in order. The Monte Carlo simulation results are very close to the analytical formulations, confirming the correctness of the proposed model and theoretical formulations.

\subsection{Impact analysis of RCS of ST}
Fig. \ref{fig2} illustrates the variation of sensing coverage probability with detection threshold across different ST's average RCS values. Additionally, the purple dashed line in Fig. \ref{fig2} represents the communication coverage probability under the same parameters. It is evident that the impact of RCS on the coverage probability is proportional. Specifically, every $10$ dB increase in RCS corresponds to a $10$ dB gain in the coverage probability. However, due to the two-way attenuation characteristic of sensing, the sensing coverage probability will not surpass the communication coverage probability, regardless of the RCS magnitude under identical conditions.

In addition, the coverage probability cannot reach $1$, no matter how small the detection threshold is. In traditional models, the integration result should be equal to $1$. However, since the blockage modeling is equivalent to the blockage of a portion of the BSs on the original distance function, Eq. \eqref{eq:fR0PDF} does not satisfy the normalization condition, which leads to a coverage hole. From the theoretical formulation, as the detection threshold approaches $0$, the coverage probability becomes an integral of the distance function $f_{R_0}(r)$, the result of which can be obtained from Eq. \eqref{eq:fR0PDFintegral}.
\begin{figure}[t]
    \centering
    \includegraphics[width=8.5cm]{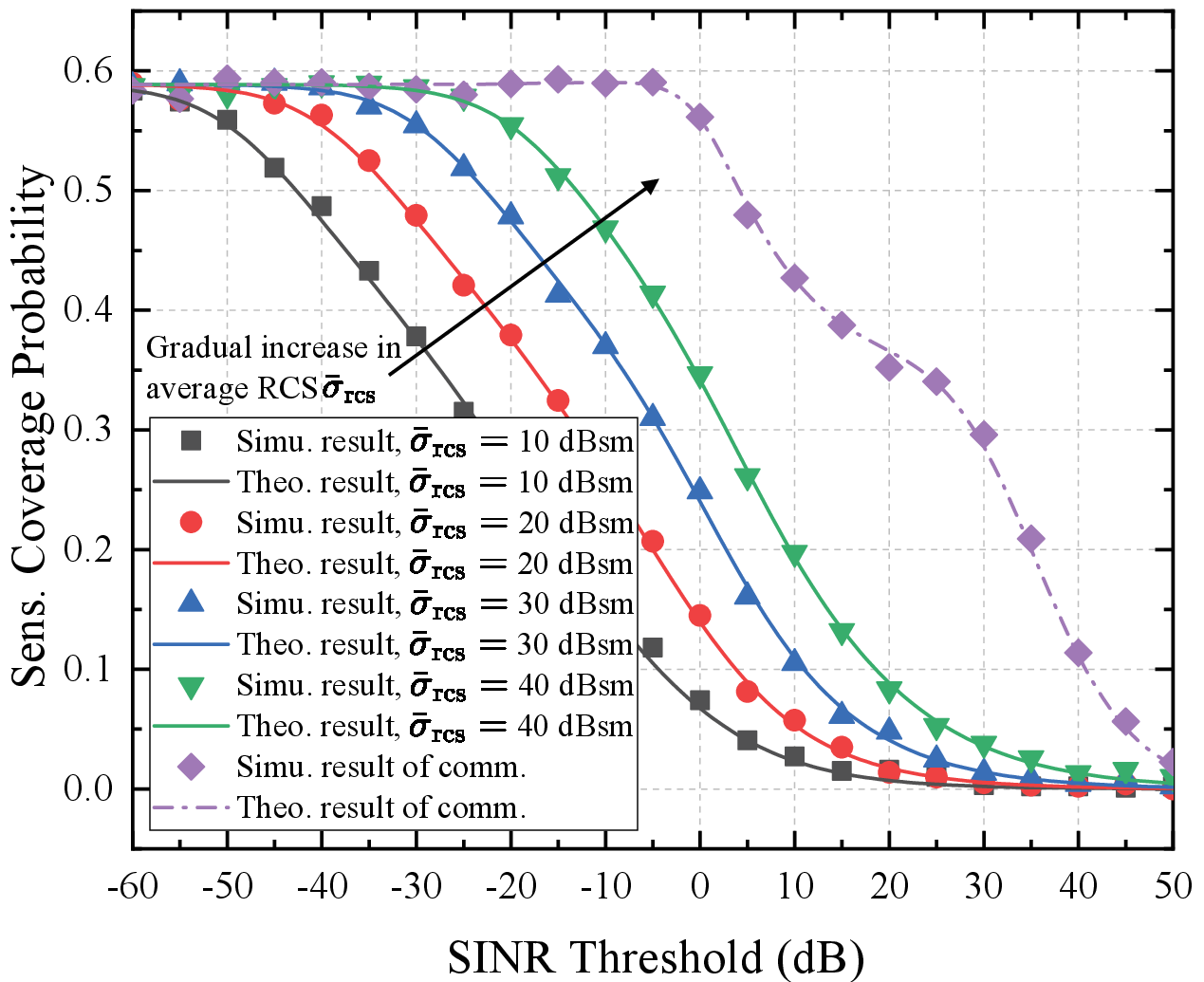}
    \caption{Variation of sensing coverage probability with detection threshold for different ST's RCS and comparison with that under communication coverage probability with the same parameters.}\label{fig2}
\end{figure}

In practical applications, the RCS of a target is not a fixed attribute, as it is subject to variations influenced by a multitude of factors, including the target's shape, surface material, angle of radar incidence, and polarization state. Nevertheless, during the system design phase, the theoretical frameworks and analytical methodologies presented herein can be leveraged, in conjunction with actual RCS data from urban environments, to conduct simulations. These simulations verify the differential effects of various design approaches on both target detection and communication capabilities. By doing so, system parameters can be preemptively optimized, thereby minimizing post-deployment tuning costs and expediting the achievement of optimal performance.

\subsection{Impact analysis of BS deployment density}
The trends of communication and sensing coverage probabilities with detection threshold for various BS deployment densities are depicted in Fig. \ref{fig3} and Fig. \ref{fig4}, revealing two unconventional phenomena.

Firstly, coverage holes only occur when the BS density falls below a certain threshold, but conversely,  at high BS densities, the coverage probability at low detection thresholds tends to approach $1$. Secondly, the lines on the graphs do not strictly follow a monotonic order, and intersections between the lines validate the varying performance of coverage probability across different BS deployment densities.

\begin{figure}[t]
    \centering
    \includegraphics[width=8.5cm]{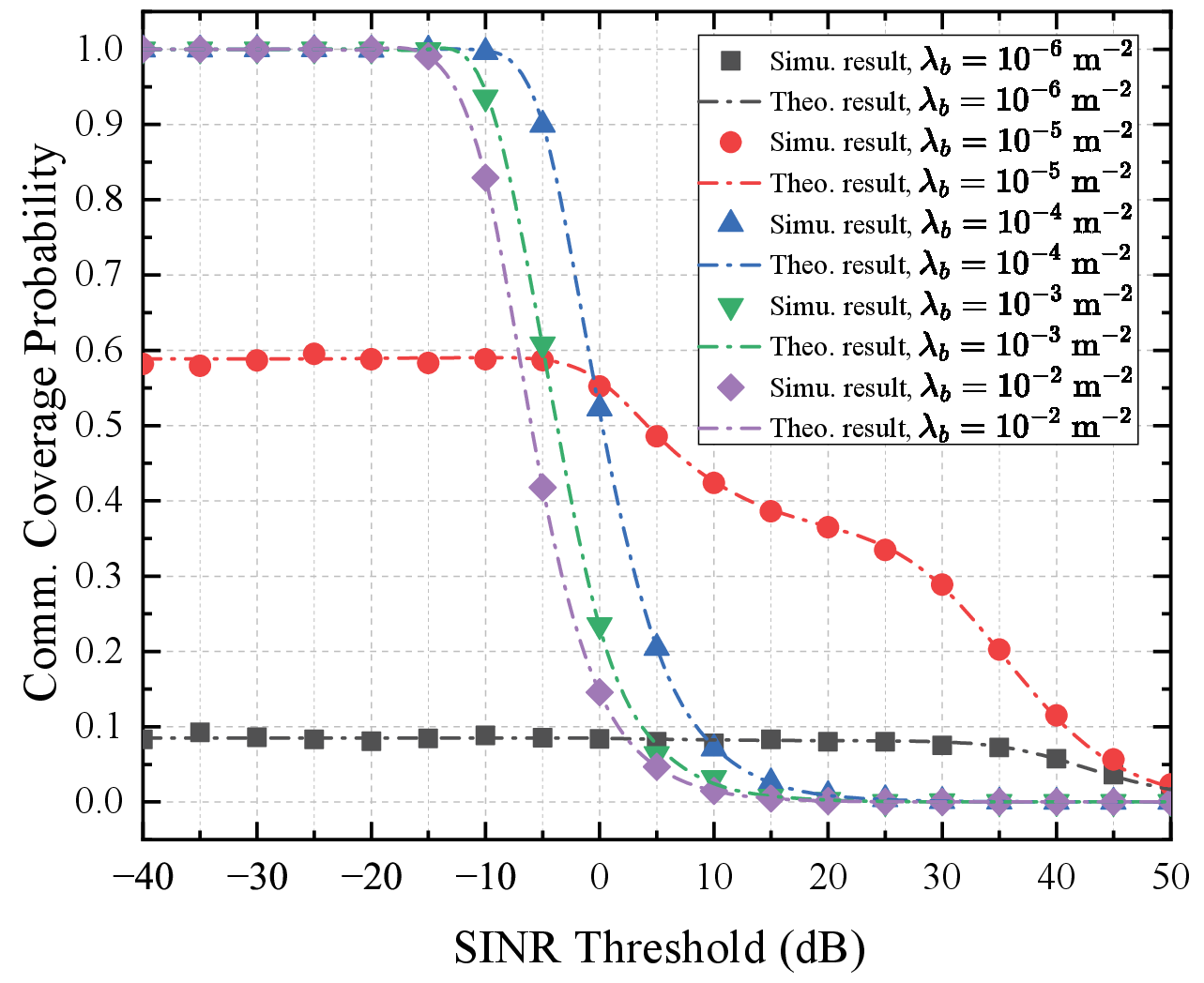}
    \caption{Variation of communication coverage probability with detection threshold for different BS deployment densities.}\label{fig3}
\end{figure}

\begin{figure}[t]
    \centering
    \includegraphics[width=8.5cm]{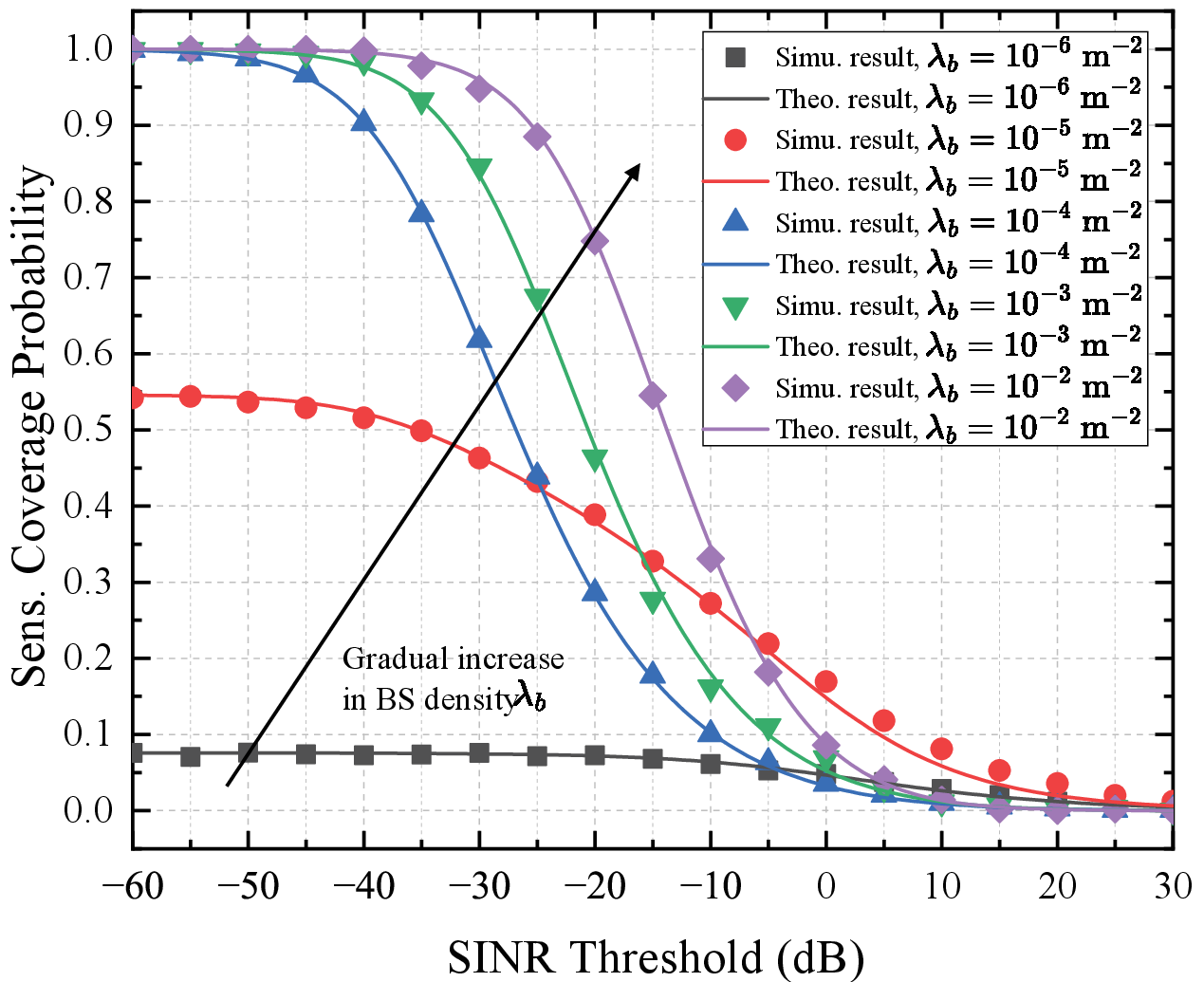}
    \caption{Variation of sensing coverage probability with detection threshold for different BS deployment densities.}\label{fig4}
\end{figure}

To further elucidate this variability, Fig. \ref{fig5} illustrates the variation of communication and sensing coverage probabilities with BS deployment density across different detection thresholds. It can be observed that the communication coverage probability demonstrates a fundamental trend of initial increase followed by decrease as the BS density changes. Due to the access strategy of the nearest visible access, as the BS density increases, the probability of users accessing closer visible BSs increases, leading to an increase in coverage probability. However, on the other hand, increasing the BS density also increases the number of visible interfering BSs. As the gain of useful signals slows down and the attenuation caused by interfering signals gradually becomes dominant, the coverage probability will decrease.

In contrast to the observed communication scenario, sensing coverage probability has two peaks, and the peak points for communication and sensing are not the same. The first peak in sensing occurs near the communication peak. The two peaks in sensing result from the competition between the useful signal and three interfering components (i.e. interference from LoS, NLoS and TRC links). However, unlike communication, the path loss of the useful signal is a two-way attenuation, so it is more slowly affected by density. Initially, as the BS density increases, the probability of accessing visible BSs increases, leading to an increase in sensing coverage probability. However, since LoS path loss is one-way, the attenuation caused by its interference becomes dominant earlier, resulting in the first peak in sensing, which occurs earlier than the peak in communication. Subsequently, the impact of LoS interference gradually diminishes, and the useful signal again becomes dominant, continuing to increase the coverage probability, resulting in the second peak in sensing coverage probability. Finally, interference from target reflection signals gradually becomes dominant, causing a slight decrease in coverage probability.

Under the current parameter settings, the main peaks of both communication and sensing are concentrated between $10^{-5}$ and $10^{-4}\;\mathrm{m}^{-2}$, which happens to be on the same order of magnitude as the blocking density. Therefore, through reasonable BS deployment and threshold settings, it is possible to simultaneously achieve good sensing coverage while ensuring communication coverage.

\begin{figure}[t]
    \centering
    \includegraphics[width=8.5cm]{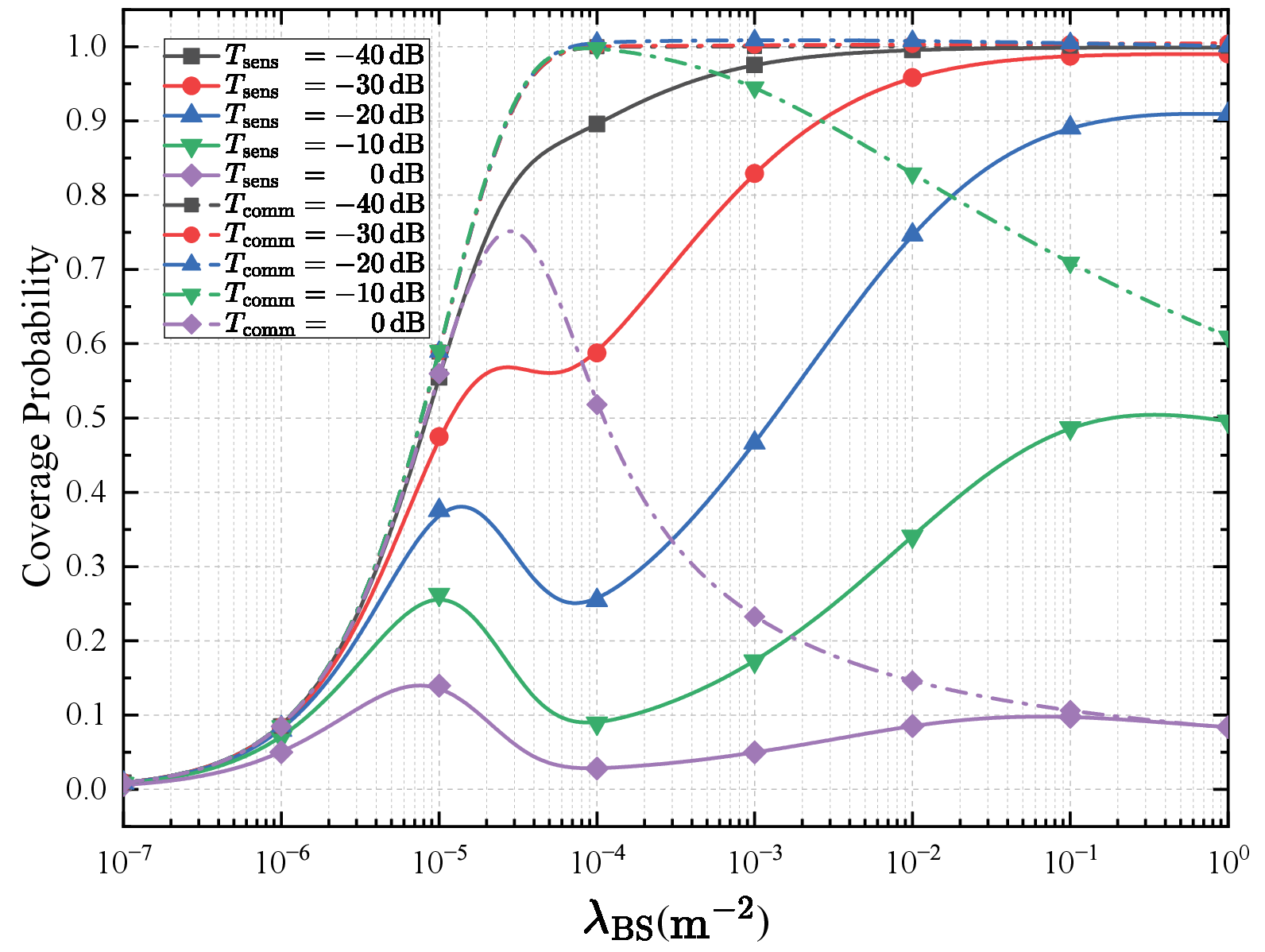}
    \caption{Variation of coverage probability with BS deployment density for different detection thresholds.}\label{fig5}
\end{figure}

\subsection{Impact analysis of environmental blockages}

Figs. \ref{fig6} and \ref{fig7} compare the variation of sensing and communication coverage probabilities with BS deployment density for blockage and non-blockage cases, respectively. Since closed-form solutions exist for Corollary 1 and Corollary 2 only when $\alpha_L>2$, the parameters are set to $\alpha_L=2.4$, $\alpha_R=4.8$ and the rest of the parameters are kept constant for the simulations. 

An intriguing observation emerges regarding the beneficial impact of blockages on both sensing and communication coverage. This phenomenon arises from that STs and UEs consistently access the nearest visible BS, resulting in the interfering link being longer than that of useful signals. Consequently, the longer the blocked link, the higher the likelihood of interference being blocked. Hence, the presence of blockages leads to a decrease in interference from LoS paths. This revelation introduces a novel perspective for the design of ISAC systems in intricate urban landscapes, suggesting that strategic manipulation of architectural layouts during the planning phase can be harnessed to boost network performance. For instance, in metropolitan settings, deliberately positioning BSs to take advantage of blockages for minimizing unwarranted interference constitutes an inventive design strategy. Furthermore, the deployment of BSs can be intelligently forecast and optimized using algorithms to harness the advantageous aspects of blockages while mitigating their detrimental effects. 

As the BS density approaches infinity, both sensing and communication coverage probabilities converge towards the non-blockage scenario, aligning with our previous conclusion. This convergence occurs because blockage was also modeled as a PPP independent from the BS. Therefore, increasing the BS density to infinity while keeping the blockage density fixed is akin to reducing the blockage density to zero. Since the impact of noise on coverage probability is minimal, Fig. \ref{fig6} already shows that communication coverage probability is independent of BS density in the non-blockage case, further confirming the conclusions of theoretical analysis.

\begin{figure}[t]
    \centering
    \includegraphics[width=8.5cm]{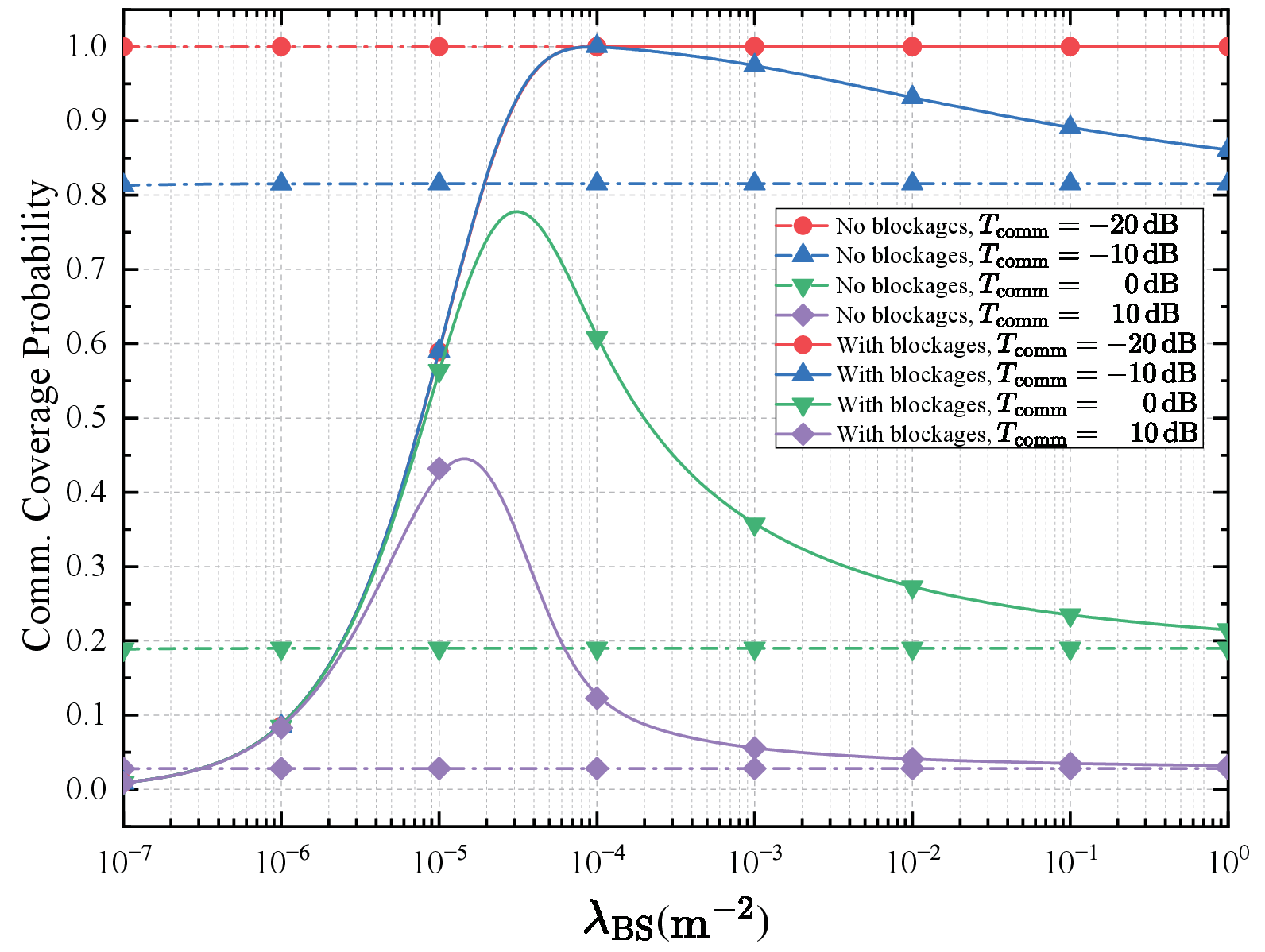}
    \caption{Variation of communication coverage probability with BS deployment density for blockage and non-blockage scenarios ($\alpha_L = 2.4$, $\alpha_N = 4.8$).}\label{fig6}
\end{figure}

\begin{figure}[t]
    \centering
    \includegraphics[width=8.5cm]{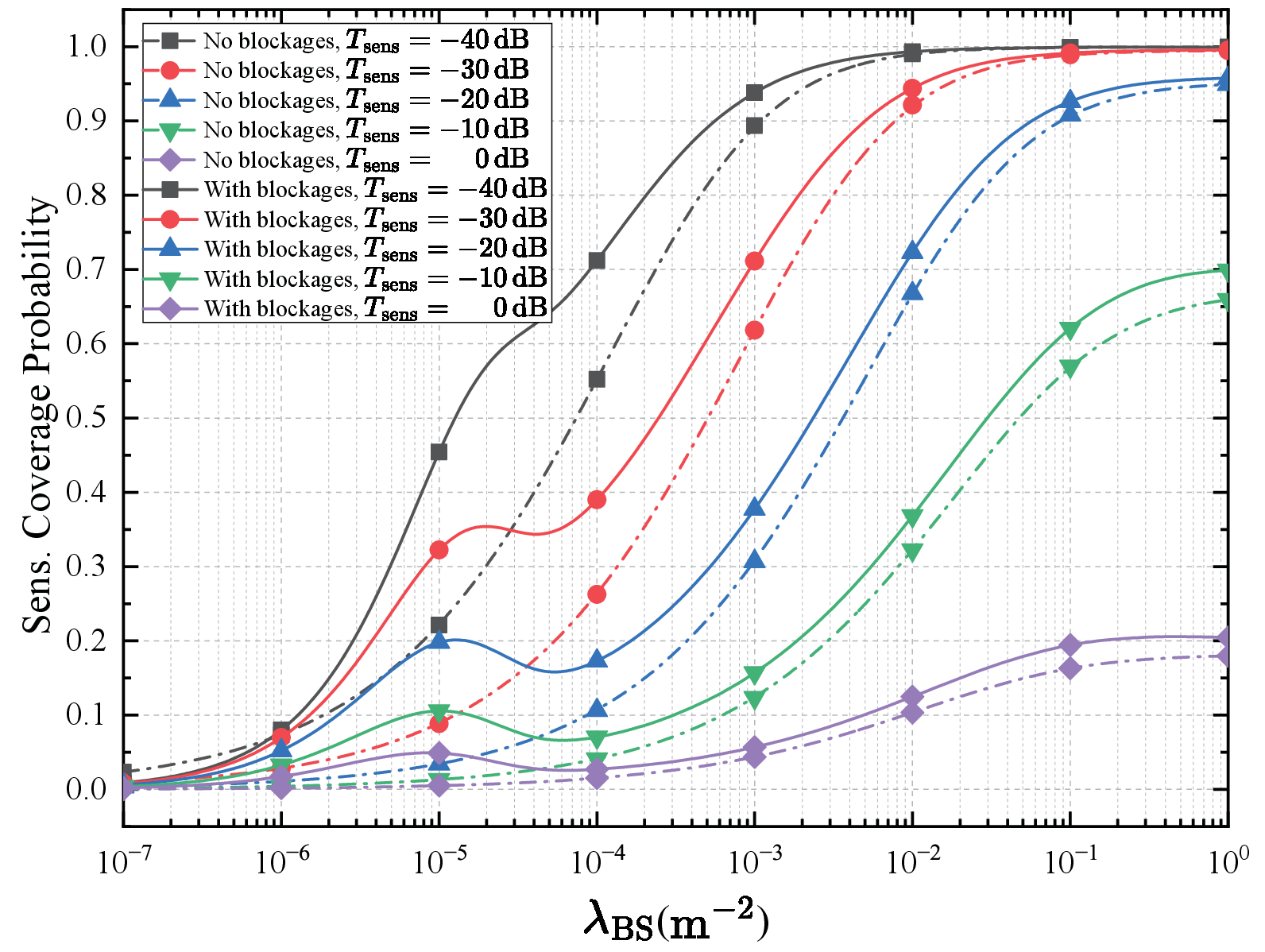}
    \caption{Variation of sensing coverage probability with BS deployment density for blockage and non-blockage scenarios ($\alpha_L = 2.4$, $\alpha_N = 4.8$).}\label{fig7}
\end{figure}

\section{Conclusion}
In this paper, a stochastic geometry approach is applied to conduct a comprehensive coverage probability performance analysis of ISAC networks, taking into account the inherent randomness in device deployment, signal propagation, target reflections, environmental blockage, and interference. Specifically, the ISAC networks in the urban scenarios is considered in the presence of blockages. Firstly, a unified and easily manageable system model is developed, which integrates the effects of LoS, NLoS, and target reflected paths. The possible blockage effects in ISAC systems are also considered, and their impacts on system performance are analyzed. Subsequently, expressions for communication and sensing SINR are separately provided based on this model, and the coverage probability is obtained as a function of BS and blockage densities, detection thresholds, and other system parameters. Lastly, special cases of no-blockage noise and no-blockage no-noise are examined. Influencing factors are isolated, and a more easily computable closed-form expression is derived. Simulation results validate the accuracy of theoretical findings and indicate that deploying more BSs is beneficial for improving communication and sensing coverage probabilities within a certain range.  An optimal BS density exists, which decreases as the detection threshold increases. Moreover, the presence of blockages produces a favorable gain in coverage probability. This approach provides valuable insights for designing and optimizing the performance of ISAC systems in real-world scenarios.

{\appendices
\section{Proof of Theorem 1}
\begin{AppendixProof}
    In this section, we give the proof of Theorem 1. According to the definition of the coverage probability by the Eq. \eqref{eq:pccommDef}, conditioning on the nearest visible BS being at a distance $R_0$ from the typical UE, the communication coverage probability can be obtained as
    \begin{equation} \label{eq:ProofPCComm1}
        \begin{aligned}
            &p_{\mathrm{c}}^{\mathrm{comm}}(T^{\mathrm{comm}},\lambda_{\mathrm{bs}},\alpha_L,\alpha_N,\beta ,p)
            \\
            &=\mathbb{E} _{R_0}[\mathbb{P} [\mathrm{SINR}^{\mathrm{comm}}>T^{\mathrm{comm}}]|R_0]
            \\
            &=\int_0^{\infty}{\mathbb{P} \left[ \frac{g_{L,0}^{\mathrm{comm}}\mathrm{PL}_L(R_0)}{I^{\mathrm{comm}}+N^{\mathrm{comm}}}>T^{\mathrm{comm}}|r \right] f_{R_0}(r)\mathrm{d}r}
            \\
            &=\int_0^{\infty}{\mathbb{P} \left[ g_{L,0}^{\mathrm{comm}}>\frac{T^{\mathrm{comm}}\left( I^{\mathrm{comm}}+N^{\mathrm{comm}} \right)}{k_Lr^{-\alpha _L}}|r \right] f_{R_0}(r)\mathrm{d}r},
        \end{aligned}
    \end{equation}
    where $g_{L,0}^{\mathrm{comm}}$ is the small-scale fading gain of the communication desired signal, and substitution into the CCDF approximation of the Rician distribution in Eq. \eqref{eq:RicianCCDF} leads to
    \begin{equation} \label{eq:ProofPCComm2}
        \begin{aligned}
            &\mathbb{P} \left[ g_{L,0}^{\mathrm{comm}}>\frac{T^{\mathrm{comm}}\left( I^{\mathrm{comm}}+N^{\mathrm{comm}} \right)}{k_Lr^{-\alpha _L}}|r \right]
            \\
            &\approx\mathbb{E} _{I^{\mathrm{comm}}}\left[ \sum_{n=1}^N{w_{n}^{K}}e^{-\frac{u_{n}^{K}r^{\alpha _L}T^{\mathrm{comm}}\left( I^{\mathrm{comm}}+N^{\mathrm{comm}} \right)}{k_L}} \right] 
            \\
            &\overset{(a)}{=}\sum_{n=1}^N{w_{n}^{K}e^{-\frac{u_{n}^{K}r^{\alpha _L}T^{\mathrm{comm}}N^{\mathrm{comm}}}{k_L}}}\mathcal{L} _{I_{\mathrm{comm}}}\left( \frac{u_{n}^{K}r^{\alpha _L}T^{\mathrm{comm}}}{k_L} \right),
        \end{aligned}
    \end{equation}
    where (a) follows from the definition of Laplace transform, and the Laplace function of $I^{\mathrm{comm}}$ then can be expressed as
    \begin{equation} \label{eq:LaplaceComm}
        \begin{aligned}
            \mathcal{L} _{I^{\mathrm{comm}}}\left( s \right) &=\int_0^{\infty}{e^{-sI^{\mathrm{comm}}}f\left( I^{\mathrm{comm}} \right)}\mathrm{d}I^{\mathrm{comm}}
            \\
            &=\mathbb{E} \left[ e^{-sI^{\mathrm{comm}}} \right] 
            \\
            &=\mathbb{E} \left[ e^{-s\left( I_{\mathrm{LoS}}^{\mathrm{comm}}+I_{\mathrm{NLoS}}^{\mathrm{comm}} \right)} \right] 
            \\
            &=\mathcal{L} _{I_{\mathrm{LoS}}^{\mathrm{comm}}}\left( s \right) \mathcal{L} _{I_{\mathrm{NLoS}}^{\mathrm{comm}}}\left( s \right).
        \end{aligned}
    \end{equation}

    Since the random variables $g^{\mathrm{comm}}_{L,i}$, $M_i$, and the point process $\Phi_{\mathrm{bs}}$ are independent of each other, the order of the expectations and the cumulative multiplication can be exchanged
    \begin{equation} \label{eq:LaplacecommLoS1}
        \begin{aligned}
            &\mathcal{L} _{I_{\mathrm{LoS}}^{\mathrm{comm}}}\left( s \right)
            \\
            &=\mathbb{E} _{\Phi_{\mathrm{bs}} ,\left\{ g_{L,i}^{\mathrm{comm}} \right\} ,\left\{ M_i \right\}}\left[ e^{-s\sum_{i\in \Phi_{\mathrm{bs}}}{g_{L,i}^{\mathrm{comm}}}\mathrm{PL}_L(R_i)M_i} \right] 
            \\
            &=\mathbb{E} _{\Phi_{\mathrm{bs}} ,\left\{ g_{L,i}^{\mathrm{comm}} \right\} ,\left\{ M_i \right\}}\left[ \prod_{i\in \Phi_{\mathrm{bs}}}{e^{-sg_{L,i}^{\mathrm{comm}}k_LR_{i}^{-\alpha _L}M_i}} \right] 
            \\
            &=\mathbb{E} _{\Phi_{\mathrm{bs}}}\Bigg[ \prod_{i\in \Phi_{\mathrm{bs}}}\mathbb{E} _{\left\{ M_i \right\}}\left[\mathbb{E} _{\left\{ g_{L,i}^{\mathrm{comm}} \right\}}\left[ e^{-sg_{L,i}^{\mathrm{comm}}k_LR_{i}^{-\alpha _L}M_i} \right] \right] 
            \\
            &\overset{\left(b \right)}{=}\mathbb{E} _{\Phi_{\mathrm{bs}}}\Bigg[ \prod_{i\in \Phi_{\mathrm{bs}}}\mathbb{E} _{\left\{ g_{L,i}^{\mathrm{comm}} \right\}}\left[ e^{-sg_{L,i}^{\mathrm{comm}}k_LR_{i}^{-\alpha _L}} \right] 
            \\
            &\qquad\qquad\qquad\quad\times\mathrm{Pr}_L(R_i)+\mathrm{Pr}_N(R_i) \Bigg],
        \end{aligned}
    \end{equation}
    where (b) uses $M_i\sim \mathrm{Bernoulli}(\mathrm{Pr}_L(R_i))$. Further, the small-scale fading of the LoS path is modeled as Rician fading, whereupon the expectation of the e-exponent with respect to $g_{L,i}^{\mathrm{comm}}$ can be calculated as
    \begin{equation} \label{eq:LaplacecommLoS2}
        \begin{aligned}
            \mathbb{E} _{\left\{ g_{L,i}^{\mathrm{comm}} \right\}}&\left[ e^{-sg_{L,i}^{\mathrm{comm}}k_LR_{i}^{-\alpha _L}}\right] 
            \\
            &=\int_0^{\infty}{e^{-sxk_L{R}_{i}^{-\alpha _L}}f_{g_{L,i}^{\mathrm{comm}}}\left( x \right)}\mathrm{d}x
            \\
            &\overset{\left( c \right)}{\approx}\sum_{n=0}^{N} w_n^K  \int_0^{\infty}{u_n^Ke^{-sxk_L{R}_{i}^{-\alpha _L}}e^{-u_n^Kx}}\mathrm{d}x
            \\
            &=\sum_{n=1}^N{w_{n}^{K}\frac{u_{n}^{K}}{u_{n}^{K}+sk_L{R}_{i}^{-\alpha _L}}},
        \end{aligned}
    \end{equation}
    where (c) is obtained using the approximate PDF of $g_{L,i}^{\mathrm{comm}}$ in Eq. \eqref{eq:RicianPDFApprox}. Substituting Eq. \eqref{eq:LaplacecommLoS2} into Eq. \eqref{eq:LaplacecommLoS1} yields that
    \begin{equation*} \label{eq:LaplacecommLoS3}
        \begin{aligned}
            &\mathcal{L} _{I_{\mathrm{LoS}}^{\mathrm{comm}}}\left( s \right)
            \\
            &\approx\mathbb{E} _{\Phi_{\mathrm{bs}}}\left[ \prod_{i\in \Phi_{\mathrm{bs}}}\left[ \sum_{n=1}^N{w_{n}^{K}\frac{u_{n}^{K}\mathrm{Pr}_L({R}_i)}{u_{n}^{K}+sk_L{R}_{i}^{-\alpha _L}}+\mathrm{Pr}_N({R}_i)}\right] \right] 
        \end{aligned}
    \end{equation*}
    \begin{equation}
        \begin{aligned}
            &\overset{\left( d \right)}{=}\exp \Bigg[ -2\pi \lambda_{\mathrm{bs}} \int_r^{\infty}x\Bigg[1 -\sum_{n=1}^N{w_{n}^{K}\frac{u_{n}^{K}\mathrm{Pr}_L(x)}{u_{n}^{K}+sk_Lx^{-\alpha _L}}} 
            \\
            &\qquad\qquad\qquad\quad-\mathrm{Pr}_N(x)\Bigg] \mathrm{d} x \Bigg]
            \\ 
            &=\exp \left[ -2\pi \lambda_{\mathrm{bs}} \sum_{n=1}^Nw_{n}^{K}\int_r^{\infty}x\left(\frac{sk_L\mathrm{Pr}_L(x)}{u_{n}^{K}x^{\alpha _L}+sk_L}\right) \mathrm{d} x \right]
            \\ 
            &=\exp \left[- 2\pi \lambda_{\mathrm{bs}} \sum_{n=1}^N{w_{n}^{K}\mathbb{F}\left( \frac{u_{n}^{K}}{sk_L},\alpha _L,\mathrm{Pr}_L(x),r  \right)}\right],
        \end{aligned}
    \end{equation}
    where (d) uses the probability generation functional  \cite{andrewsTractableApproachCoverage2011} (PGFL) of PPP, which states for some function $f(x)$ that
    \begin{equation} \label{eq:PGFL}
        \begin{aligned}
            \mathbb{E} _{\Phi_{\mathrm{bs}}}\left[ \prod_{i\in \Phi_{\mathrm{bs}}}{f\left( x \right)} \right] =\exp \left[ -\int_{\mathbb{R} ^2}{\left( 1-f\left( x \right) \right) \lambda_{\mathrm{bs}} \left( \mathrm{d}x \right)} \right].
        \end{aligned}
    \end{equation}
    
    Since the associated BS is the nearest visible BS and there will be no other visible BSs closer than the associated BS, the integral starts at $r$ when using PGFL. The last step in the derivation defines a new function $\mathbb{F}(\cdot)$ which can be expressed as
    \begin{equation} \label{eq:F}
        \begin{aligned}
            \mathbb{F}\left( \varepsilon ,\alpha ,p(x),h \right) =\int_h^{\infty}{\frac{xp(x)}{\varepsilon x^{\alpha}+1}\mathrm{d}x}.
        \end{aligned}
    \end{equation}

    Similarly, for the Laplace transform of the NLoS interference, since the random variables $g^{\mathrm{comm}}_{N,i}$, ${S}_i$, and the point process $\Phi_{\mathrm{bs}}$ are independent of each other, we have
    \begin{equation} \label{eq:LaplacecommNLoS1}
        \begin{aligned}
            &\mathcal{L} _{I_{\mathrm{NLoS}}^{\mathrm{comm}}}\left( s \right) 
            \\
            &=\mathbb{E} _{\Phi_{\mathrm{bs}} ,\left\{ g_{N,i}^{\mathrm{comm}} \right\} ,\left\{ M_i \right\}}\left[ e^{-s\sum_{i\in \Phi_{\mathrm{bs}}}{g_{N,i}^{\mathrm{comm}}}\mathrm{PL}_N(R_i)(1-M_i)} \right] 
            \\
            &=\mathbb{E} _{\Phi_{\mathrm{bs}} ,\left\{ g_{N,i}^{\mathrm{comm}} \right\} ,\left\{ M_i \right\}}\left[ \prod_{i\in \Phi_{\mathrm{bs}}}{e^{-sg_{N,i}^{\mathrm{comm}}k_NR_{i}^{-\alpha _N}(1-M_i)}} \right] 
            \\
            &=\mathbb{E} _{\Phi_{\mathrm{bs}}}\Bigg[ \prod_{i\in \Phi_{\mathrm{bs}}}\mathbb{E} _{\left\{ M_i \right\}}\left[\mathbb{E} _{\left\{ g_{N,i}^{\mathrm{comm}} \right\}}\left[ e^{-sg_{N,i}^{\mathrm{comm}}k_NR_{i}^{-\alpha _N}(1-M_i)} \right] \right] 
            \\
            &\overset{\left( e \right)}{=}\mathbb{E} _{\Phi_{\mathrm{bs}}}\Bigg[ \prod_{i\in \Phi_{\mathrm{bs}}}\Big[\mathbb{E} _{\left\{ g_{N,i}^{\mathrm{comm}} \right\}}\left[ e^{-sg_{N,i}^{\mathrm{comm}}k_NR_{i}^{-\alpha _N}} \right] 
            \\
            &\qquad\qquad\qquad\quad\times\mathrm{Pr}_N(R_i) +\mathrm{Pr}_L(R_i) \Big]\Bigg] ,
        \end{aligned}
    \end{equation}
    where (e) uses $M_i\sim \mathrm{Bernoulli}(\mathrm{Pr}_L(R_i))$. The small-scale fading of the NLoS path is modeled as Rayleigh fading, after that the expectation with respect to $g_{N,i}^{\mathrm{comm}}$ can be formulated as follows
    \begin{equation} \label{eq:LaplacecommNLoS2}
        \begin{aligned}
            \mathbb{E} _{\left\{ g_{N,i}^{\mathrm{comm}} \right\}}&\left[ e^{-sg_{N,i}^{\mathrm{comm}}k_NR_{i}^{-\alpha _N}}\right]
            \\
            &=\int_0^{\infty}{e^{-sxk_N{R}_{i}^{-\alpha _N}}f_{g_{N,i}^{\mathrm{comm}}}\left( x \right)}\mathrm{d}x
            \\
            &\overset{\left( f \right)}{=}\mu_{N}^{\mathrm{comm}}\int_0^{\infty}{e^{-x\left( \mu_{N}^{\mathrm{comm}}+{sk_NR_{i}^{-\alpha _N}} \right)}}\mathrm{d}x
            \\
            &=\frac{\mu_{N}^{\mathrm{comm}}}{\mu_{N}^{\mathrm{comm}}+sk_NR_{i}^{-\alpha _N}},
        \end{aligned}
    \end{equation}
    where (f) holds because Eq. \eqref{eq:RayleighPDF} and $\mu^{\mathrm{comm}}_{N}$ is the Rayleigh fading parameters for the communication NLoS interference path. Hence,
    \begin{equation} \label{eq:LaplacecommNLoS3}
        \begin{aligned}
            &\mathcal{L} _{I_{\mathrm{NLoS}}^{\mathrm{comm}}}\left( s \right) 
            \\
            &=\mathbb{E} _{\Phi_{\mathrm{bs}}}\left[ \prod_{i\in \Phi_{\mathrm{bs}}}{\left[ \frac{\mu_{N}^{\mathrm{comm}}\mathrm{Pr}_N(R_i)}{\mu_{N}^{\mathrm{comm}}+sk_NR_{i}^{-\alpha _N}}+\mathrm{Pr}_L(R_i) \right]} \right] 
            \\
            &\overset{\left( g \right)}{=}\exp \left[ -2\pi \lambda_{\mathrm{bs}} \int_0^{\infty}{x\left( \frac{sk_N\mathrm{Pr}_N(x) }{\mu_{N}^{\mathrm{comm}}x^{\alpha _N}+sk_N} \right) \mathrm{d}x} \right] 
            \\
            &=\exp \left[ -2\pi \lambda_{\mathrm{bs}} \mathbb{F}\left( \frac{\mu_{N}^{\mathrm{comm}}}{sk_N},\alpha _N,\mathrm{Pr}_N(x),0 \right) \right],
        \end{aligned}
    \end{equation}
    where (g) uses the PGFL, but the integration starts from 0 because there may be other BSs without direct paths within a circle of radius $r$ under the nearest visible access policy.

    Multiply the results of \eqref{eq:LaplacecommLoS3} and \eqref{eq:LaplacecommNLoS3}, then substitute them into \eqref{eq:ProofPCComm2} and then into \eqref{eq:ProofPCComm1} yields the proof.
\end{AppendixProof}

\section{Proof of Theorem 2}
\begin{AppendixProof}
    In this section, we give the proof of Theorem 2. According to the definition of the coverage probability by the Eq. \eqref{eq:pcsensDef}, conditioning on the nearest visible BS being at a distance $\tilde{R}_0$ from the typical ST, the sensing coverage probability can be calculated as
    \begin{equation} \label{eq:ProofPCsens}
        \begin{aligned}
            &p_{\mathrm{c}}^{\mathrm{sens}}(T^{\mathrm{sens}},\lambda_{\mathrm{bs}},\alpha_L,\alpha_N,\alpha_R,\beta ,p)
            \\
            &=\mathbb{E} _{\tilde{R}_0}[\mathbb{P} [\mathrm{SINR}^{\mathrm{sens}}>T^{\mathrm{sens}}]|\tilde{R}_0]
            \\
            &=\int_0^{\infty}{\mathbb{P}}[\mathrm{SINR}^{\mathrm{sens}}>T^{\mathrm{sens}}]f_{\tilde{R}_0}(r)\mathrm{d}r
            \\
            &=\int_0^{\infty}{\mathbb{P} \left[ \frac{\sigma _{\mathrm{rcs}}\mathrm{PL}_R(\tilde{R}_0)}{I^{\mathrm{sens}}+N^{\mathrm{sens}}}>T^{\mathrm{sens}}|r \right] f_{\tilde{R}_0}(r)\mathrm{d}r}
            \\
            &=\int_0^{\infty}{\mathbb{P} \left[ \sigma _{\mathrm{rcs}}>\frac{T^{\mathrm{sens}}\left( I^{\mathrm{sens}}+N^{\mathrm{sens}} \right)}{k_Rr^{-\alpha _R}}|r \right] f_{\tilde{R}_0}(r)\mathrm{d}r}
            \\
            &\overset{(a)}{=}\int_0^{\infty}{\mathbb{E} _{I^{\mathrm{sens}}}\left[ e^{-\frac{r^{\alpha _R}T^{\mathrm{sens}}\left( I^{\mathrm{sens}}+N^{\mathrm{sens}} \right)}{\bar{\sigma}_{\mathrm{rcs}}k_R}} \right] f_{\tilde{R}_0}(r)\mathrm{d}r}
            \\
            &\overset{(b)}{=}\int_0^{\infty}{e^{-\frac{r^{\alpha _R}T^{\mathrm{sens}}N^{\mathrm{sens}}}{\bar{\sigma}_{\mathrm{rcs}}k_R}}\mathcal{L} _{I^{\mathrm{sens}}}\left( \frac{r^{\alpha _R}T^{\mathrm{sens}}}{\bar{\sigma}_{\mathrm{rcs}}k_R} \right) f_{\tilde{R}_0}(r)\mathrm{d}r},
        \end{aligned}
    \end{equation}
    where (a) follows from the CCDF of the RCS in Eq. \eqref{eq:rcsCCDF}, (b) follows from the definition of Laplace transform. With the considered interference model in Eq. \eqref{eq:Isens}, the Laplace function of $I^{\mathrm{sens}}$ then can be expressed as
    \begin{equation} \label{eq:LaplaceSens}
        \begin{aligned}
            \mathcal{L} _{I^{\mathrm{sens}}}\left( s \right) &=\int_0^{\infty}{e^{-sI^{\mathrm{sens}}}f\left( I^{\mathrm{sens}} \right)}\mathrm{d}I^{\mathrm{sens}}
            \\
            &=\mathbb{E} \left[ e^{-sI^{\mathrm{sens}}} \right] 
            \\
            &=\mathbb{E} \left[ e^{-s\left( I_{\mathrm{LoS}}^{\mathrm{sens}}+I_{\mathrm{NLoS}}^{\mathrm{sens}}+I_{\mathrm{TRC}}^{\mathrm{sens}} \right)} \right] 
            \\
            &=\mathcal{L} _{I_{\mathrm{LoS}}^{\mathrm{sens}}}\left( s \right) \mathcal{L} _{I_{\mathrm{NLoS}}^{\mathrm{sens}}}\left( s \right) \mathcal{L} _{I_{\mathrm{TRC}}^{\mathrm{sens}}}\left( s \right) .
        \end{aligned}
    \end{equation}

    Since the random variables $g^{\mathrm{sens}}_{L,i}$, $\hat{M}_i$ and the point process $\Phi_{\mathrm{bs}}$ are independent of each other, similar to the proof process for Theorem 1, we have
    \begin{equation} \label{eq:LaplaceSensLoS1}
        \begin{aligned}
            &\mathcal{L} _{I_{\mathrm{LoS}}^{\mathrm{sens}}}\left( s \right)
            \\ 
            &=\mathbb{E} _{\Phi_{\mathrm{bs}} ,\left\{ g_{L,i}^{\mathrm{sens}} \right\} ,\left\{ \hat{M}_i \right\}}\left[ e^{-s\sum_{i\in \Phi_{\mathrm{bs}}}{g_{L,i}^{\mathrm{sens}}}\mathrm{PL}_L(\hat{R}_i)\hat{M}_i} \right] 
            \\
            &=\mathbb{E} _{\Phi_{\mathrm{bs}} ,\left\{ g_{L,i}^{\mathrm{sens}} \right\} ,\left\{ \hat{M}_i \right\}}\left[ \prod_{i\in \Phi_{\mathrm{bs}}}{e^{-sg_{L,i}^{\mathrm{sens}}k_L\hat{R}_{i}^{-\alpha _L}\hat{M}_i}} \right] 
            \\
            &\overset{(c)}{=}\mathbb{E} _{\Phi_{\mathrm{bs}}}\Bigg[ \prod_{i\in \Phi_{\mathrm{bs}}}\Big[\mathbb{E} _{\left\{ g_{L,i}^{\mathrm{sens}} \right\}}\left[ e^{-sg_{L,i}^{\mathrm{sens}}k_L\hat{R}_{i}^{-\alpha _L}} \right] 
            \\
            &\qquad\qquad\qquad\quad\times \mathrm{Pr}_L(\hat{R}_i) + \mathrm{Pr}_N(\hat{R}_i) \Big]\Bigg] ,
        \end{aligned}
    \end{equation}
    where (c) uses $\hat{M}_i\sim \mathrm{Bernoulli}(\mathrm{Pr}_L(\hat{R}_i))$. The small-scale fading of the LoS path is modeled as Rician fading, so the expectation in Eq. \eqref{eq:LaplaceSensLoS1} with respect to $g_{L,i}^{\mathrm{sens}}$ can be computed as
    \begin{equation} \label{eq:LaplaceSensLoS2}
        \begin{aligned}
            \mathbb{E} _{\left\{ g_{L,i}^{\mathrm{sens}} \right\}}&\left[ e^{-sg_{L,i}^{\mathrm{sens}}k_L\hat{R}_{i}^{-\alpha _L}}\right] 
            \\
            &=\int_0^{\infty}{e^{-sxk_L\hat{R}_{i}^{-\alpha _L}}f_{g_{L,i}^{\mathrm{sens}}}\left( x \right)}\mathrm{d}x
            \\
            &\overset{\left( d \right)}{\approx}\sum_{n=0}^{N} w_n^K  \int_0^{\infty}{u_n^Ke^{-sxk_L\hat{R}_{i}^{-\alpha _L}}e^{-u_n^Kx}}\mathrm{d}x
            \\
            &=\sum_{n=1}^N{w_{n}^{K}\frac{u_{n}^{K}}{u_{n}^{K}+sk_L\hat{R}_{i}^{-\alpha _L}}},
        \end{aligned}
    \end{equation}
    where (d) is obtained using the approximate PDF of $g_{L,i}^{\mathrm{sens}}$ in Eq. \eqref{eq:RicianPDFApprox}. Substituting Eq. \eqref{eq:LaplaceSensLoS2} into Eq. \eqref{eq:LaplaceSensLoS1} yields that
    \begin{equation} \label{eq:LaplacesensLoS3}
        \begin{aligned}
            &\mathcal{L} _{I_{\mathrm{LoS}}^{\mathrm{sens}}}\left( s \right)
            \\
            &\approx\mathbb{E} _{\Phi_{\mathrm{bs}}}\left[ \prod_{i\in \Phi_{\mathrm{bs}}}\left[ \sum_{n=1}^N{w_{n}^{K}\frac{u_{n}^{K}\mathrm{Pr}_L(\hat{R}_i)}{u_{n}^{K}+sk_L\hat{R}_{i}^{-\alpha _L}}+\mathrm{Pr}_N(\hat{R}_i)}\right] \right] 
            \\
            &\overset{\left( e \right)}{=}\exp \Bigg[ -2\pi \lambda_{\mathrm{bs}} \int_r^{\infty}x\Bigg[1 -\sum_{n=1}^N{w_{n}^{K}\frac{u_{n}^{K}\mathrm{Pr}_L(x)}{u_{n}^{K}+sk_Lx^{-\alpha _L}}} 
            \\
            &\qquad\qquad\qquad\quad-\mathrm{Pr}_N(x)\Bigg] \mathrm{d} x \Bigg]
            \\ 
            &=\exp \left[ -2\pi \lambda_{\mathrm{bs}} \sum_{n=1}^Nw_{n}^{K}\int_r^{\infty}x\left(\frac{sk_L\mathrm{Pr}_L(x)}{u_{n}^{K}x^{\alpha _L}+sk_L}\right) \mathrm{d} x \right]
            \\ 
            &=\exp \left[- 2\pi \lambda_{\mathrm{bs}} \sum_{n=1}^N{w_{n}^{K}\mathbb{F}\left( \frac{u_{n}^{K}}{sk_L},\alpha _L,\mathrm{Pr}_L(x),r  \right)}\right],
        \end{aligned}
    \end{equation}
    where (e) uses the PGFL of PPP, $\mathbb{F}\left( \varepsilon ,\alpha ,p(x),h \right) =\int_h^{\infty}{\frac{xp(x)}{\varepsilon x^{\alpha}+1}\mathrm{d}x}$. Note that when using PGFL, the lower bound of the integral starts at $r$ rather than at $0$. This is because sensing is associated with the nearest visible BS, and no other visible BS will exist within a circle with radius $r$ and center ST. This fraction needs to be shaved from the plane when calculating LoS interference.

    Similarly, Laplace transform of NLoS interference links can be derived as
    \begin{equation} \label{eq:LaplaceSensNLoS}
        \begin{aligned}
            &\mathcal{L} _{I_{\mathrm{NLoS}}^{\mathrm{sens}}}\left( s \right) 
            \\
            &=\mathbb{E} _{\Phi_{\mathrm{bs}} ,\left\{ g_{N,i}^{\mathrm{sens}} \right\} ,\left\{ \hat{M}_i \right\}}\left[ e^{-s\sum_{i\in \Phi_{\mathrm{bs}}}{g_{N,i}^{\mathrm{sens}}}\mathrm{PL}_N(\hat{R}_i)(1-\hat{M}_i)} \right] 
            \\
            &=\mathbb{E} _{\Phi_{\mathrm{bs}} ,\left\{ g_{N,i}^{\mathrm{sens}} \right\} ,\left\{ \hat{M}_i \right\}}\left[ \prod_{i\in \Phi_{\mathrm{bs}}}{e^{-sg_{N,i}^{\mathrm{sens}}k_N\hat{R}_{i}^{-\alpha _N}(1-\hat{M}_i)}} \right] 
            \\
            &\overset{\left( f \right)}{=}\mathbb{E} _{\Phi_{\mathrm{bs}}}\Bigg[ \prod_{i\in \Phi_{\mathrm{bs}}}\Big[\mathbb{E} _{\left\{ g_{N,i}^{\mathrm{sens}} \right\}}\left[ e^{-sg_{N,i}^{\mathrm{sens}}k_N\hat{R}_{i}^{-\alpha _N}} \right] 
            \\
            &\qquad\qquad\qquad\quad\times \mathrm{Pr}_N(\hat{R}_i)  +\mathrm{Pr}_L(\hat{R}_i) \Big]\Bigg] ,
        \end{aligned}
    \end{equation}
    where (f) holds because $\hat{M}_i\sim \mathrm{Bernoulli}(\mathrm{Pr}_L(\hat{R}_i))$, and $g^{\mathrm{sens}}_{N,i}$, $\hat{M}_i$ and the point process $\Phi_{\mathrm{bs}}$ are independent of each other. The expectation with respect to $g_{N,i}^{\mathrm{sens}}$ can be formulated as follows
    \begin{equation} \label{eq:LaplacesensNLoS2}
        \begin{aligned}
            \mathbb{E} _{\left\{ g_{N,i}^{\mathrm{sens}} \right\}}&\left[ e^{-sg_{N,i}^{\mathrm{sens}}k_NR_{i}^{-\alpha _N}}\right]
            \\
            &=\int_0^{\infty}{e^{-sxk_N{R}_{i}^{-\alpha _N}}f_{g_{N,i}^{\mathrm{sens}}}\left( x \right)}\mathrm{d}x
            \\
            &\overset{\left(g \right)}{=}\mu_{N}^{\mathrm{sens}}\int_0^{\infty}{e^{-x\left( \mu_{N}^{\mathrm{sens}}+{sk_NR_{i}^{-\alpha _N}} \right)}}\mathrm{d}x
            \\
            &=\frac{\mu_{N}^{\mathrm{sens}}}{\mu_{N}^{\mathrm{sens}}+sk_NR_{i}^{-\alpha _N}},
        \end{aligned}
    \end{equation}
    where (g) holds because the the PDF of $g_{N,i}^{\mathrm{sens}} $ in Eq. \eqref{eq:RayleighPDF}. Hence,
    \begin{equation} \label{eq:LaplaceSensNLoS3}
        \begin{aligned}
            &\mathcal{L} _{I_{\mathrm{NLoS}}^{\mathrm{sens}}}\left( s \right) 
            \\
            &=\mathbb{E} _{\Phi_{\mathrm{bs}}}\left[ \prod_{i\in \Phi_{\mathrm{bs}}}{\left[ \frac{\mu_{N}^{\mathrm{sens}}\mathrm{Pr}_N(\hat{R}_i)}{\mu_{N}^{\mathrm{sens}}+sk_N\hat{R}_{i}^{-\alpha _N}}+\mathrm{Pr}_L(\hat{R}_i) \right]} \right] 
            \\
            &\overset{\left( h \right)}{=}\exp \left[ -2\pi \lambda_{\mathrm{bs}} \int_0^{\infty}{x\left( \frac{sk_N\mathrm{Pr}_N(x)}{\mu_{N}^{\mathrm{sens}}x^{\alpha _N}+sk_N} \right) \mathrm{d}x} \right]
            \\
            &=\exp \left[ -2\pi \lambda_{\mathrm{bs}} \mathbb{F}\left( \frac{\mu_{N}^{\mathrm{sens}}}{sk_N},\alpha _N,\mathrm{Pr}_N(x),0 \right) \right] ,
        \end{aligned}
    \end{equation}
    where (h) uses the PGFL. The integration starts from $0$ because the association policy is only relevant for BSs where LoS is present and does not affect NLoS BSs.

    In the same vein,
    \begin{equation*}
        \begin{aligned}\label{eq:LaplaceSensObj}
            &\mathcal{L} _{I_{\mathrm{TRC}}^{\mathrm{sens}}}\left( s \right) 
            \\
            &=\mathbb{E} _{\Phi_{\mathrm{bs}} ,\sigma _{\mathrm{rcs}},\left\{ \tilde{M}_i \right\}}\left[ e^{-s\sum_{i\in \Phi_{\mathrm{bs}}}{\sigma _{\mathrm{rcs}}k_R\tilde{R}_{i}^{-\alpha _L}\tilde{M}_i\times \tilde{R}_{0}^{-\alpha _L}\tilde{M}_0}} \right] 
        \end{aligned}
    \end{equation*}

    \begin{equation}
        \begin{aligned}
            &=\mathbb{E} _{\Phi_{\mathrm{bs}} ,\sigma _{\mathrm{rcs}},\left\{ \tilde{M}_i \right\}}\left[ \prod_{i\in \Phi_{\mathrm{bs}}}{e^{-s\sigma _{\mathrm{rcs}}k_R\tilde{R}_{i}^{-\alpha _L}\tilde{M}_i\times \tilde{R}_{0}^{-\alpha _L}\tilde{M}_0}} \right]
            \\
            &\overset{\left( i \right)}{=}\mathbb{E} _{\Phi_{\mathrm{bs}}}\Bigg[ \prod_{i\in \Phi_{\mathrm{bs}}}\bigg[ \mathbb{E} _{\sigma _{\mathrm{rcs}}}\left[ e^{-sk_R\sigma _{\mathrm{rcs}}\tilde{R}_{i}^{-\alpha _L}\tilde{R}_{0}^{-\alpha _L}} \right] \mathrm{Pr}_L(\tilde{R}_i)
            \\
            &\qquad\qquad\qquad\quad+\mathrm{Pr}_N(\tilde{R}_i) \bigg] \Bigg] 
            \\
            &\overset{\left( j \right)}{=}\mathbb{E} _{\Phi_{\mathrm{bs}}}\left[ \prod_{i\in \Phi_{\mathrm{bs}}}{\left[ \frac{\mathrm{Pr}_L(\tilde{R}_i)}{1+sk_R\bar{\sigma}_{\mathrm{rcs}}\tilde{R}_{i}^{-\alpha _L}\tilde{R}_{0}^{-\alpha _L}}+\mathrm{Pr}_N(\tilde{R}_i) \right]} \right] 
            \\
            &\overset{\left( k \right)}{=}\exp \left[ -2\pi \lambda_{\mathrm{bs}} \int_r^{\infty}{x\left(\frac{sk_R\mathrm{Pr}_L(x)}{\bar{\sigma}_{\mathrm{rcs}}^{-1}x^{\alpha _L}r^{\alpha _L}+sk_R}\right)\mathrm{d}x} \right]  
            \\
            &=\exp \left[ -2\pi \lambda_{\mathrm{bs}} \mathbb{F}\left( \frac{r^{\alpha _L}}{sk_R\bar{\sigma}_{\mathrm{rcs}}},\alpha _L,\mathrm{Pr}_L(x),r \right) \right] ,
        \end{aligned}
    \end{equation}
    where (i) holds because $\tilde{M}_i\sim \mathrm{Bernoulli}(\mathrm{Pr}_L(\tilde{R}_i))$ and $\tilde{M}_0=1$, and (j) uses the CCDF of RCS in Eq \eqref{eq:rcsCCDF}. Since the sensing association strategy is the nearest visible BS, no other visible BS will exist in the circle centered on ST with radius $r$. Therefore, when PGFL is used in step (k), the integral is calculated from $r$.

    Plugging \eqref{eq:LaplacesensLoS3}, \eqref{eq:LaplaceSensNLoS3} and \eqref{eq:LaplaceSensObj} into \eqref{eq:LaplaceSens} and then into \eqref{eq:ProofPCsens} gives the desired result.
\end{AppendixProof}
}

\bibliographystyle{IEEEtran}

\bibliography{IEEEabrv,./References.bib}
\begin{IEEEbiography}[{\includegraphics[width=1in,height=1.25in,clip,keepaspectratio]{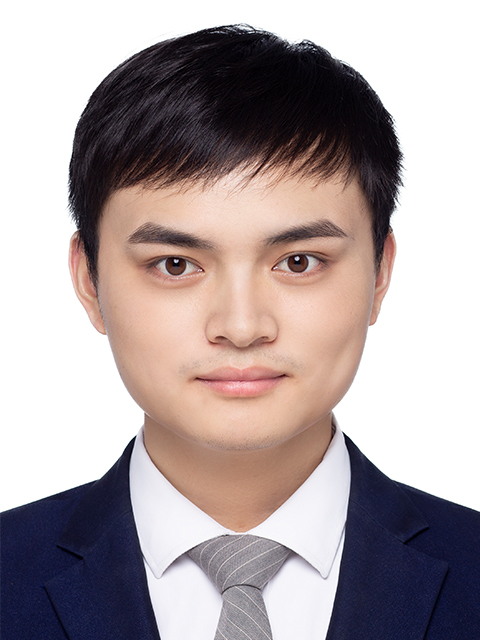}}]{Zezhong Sun}
        (Student Member, IEEE) received the B.S. degree in information engineering from Beijing University of Posts and Telecommunications (BUPT), Beijing, China, in 2019. He is currently working toward the Ph.D. degree in Information and Communication Engineering at the State Key Laboratory of Networking and Switching Technology (SKT-NST), BUPT. His research interests include information theory on integrated sensing and communications (ISAC), stochastic geometry and wireless sensing.
\end{IEEEbiography}

\begin{IEEEbiography}[{\includegraphics[width=1in,height=1.25in,clip,keepaspectratio]{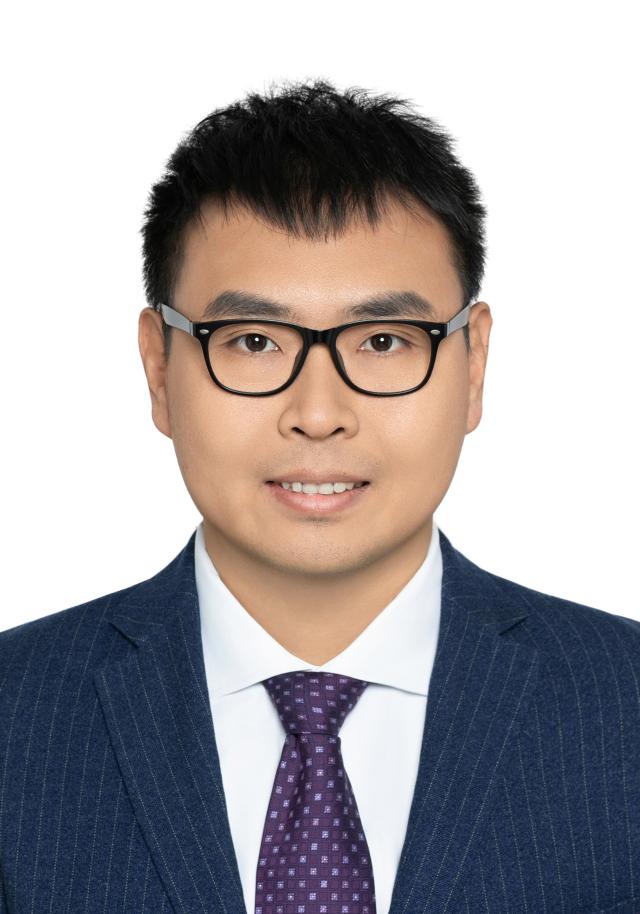}}]{Shi Yan}
        (Member, IEEE) received the Ph.D. degree in communication and information engineering from Beijing University of Posts and Telecommunications (BUPT), Beijing, China, in 2017. He is currently an Associate Professor with the State Key Laboratory of Networking and Switching Technology (SKT-NST), BUPT. In 2015, he was an Academic Visiting Scholar with Arizona State University, Tempe, AZ, USA. His research interests include game theory, integrated sensing and communication (ISAC), resource management, deep reinforcement learning, stochastic geometry, and fog radio access networks.
\end{IEEEbiography}

\begin{IEEEbiography}[{\includegraphics[width=1in,height=1.25in,clip,keepaspectratio]{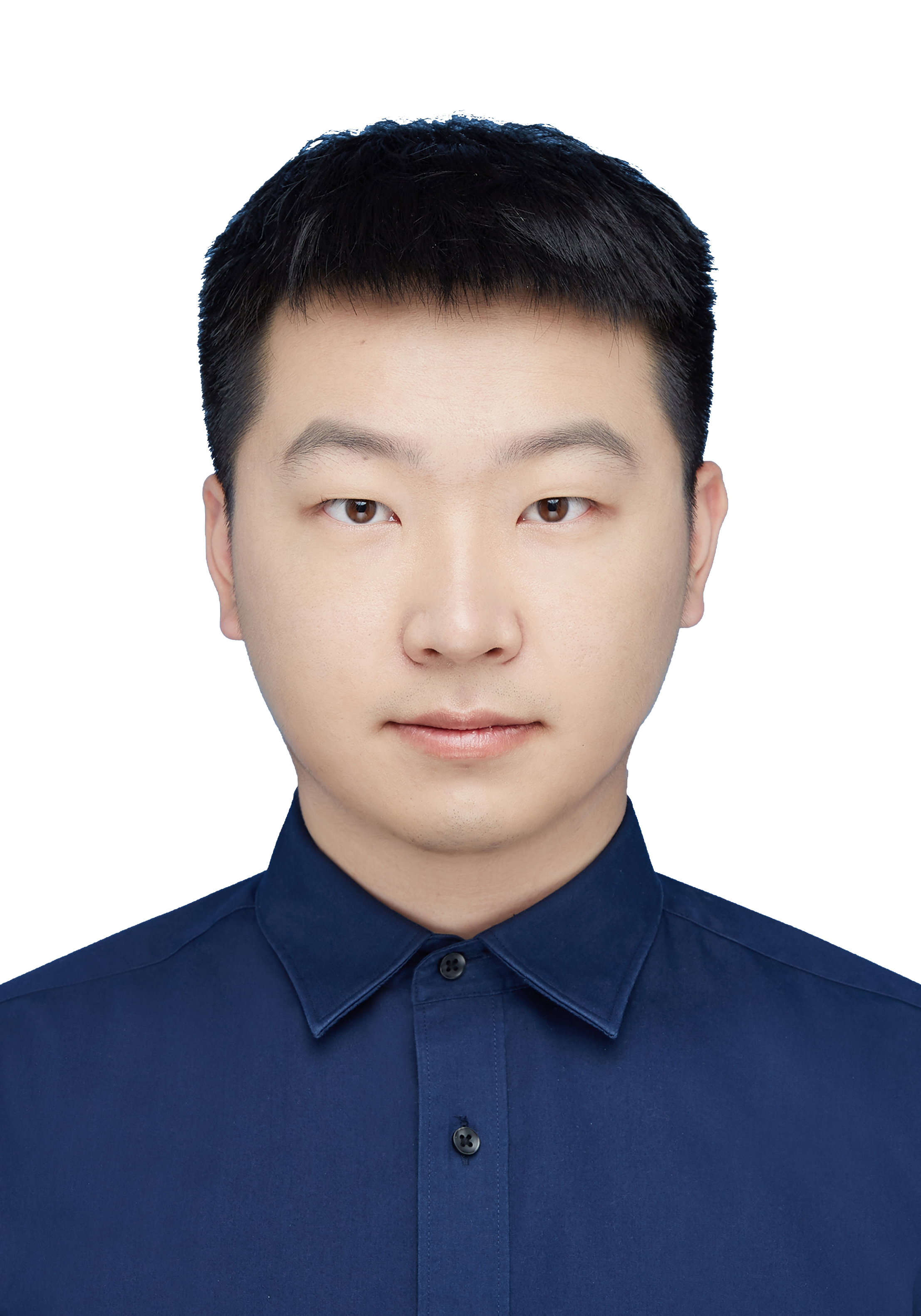}}]{Ning Jiang}
        (Student Member, IEEE) received the B.S. degree and M.S. degree in information and communication engineering from the Beijing University of Posts and Telecommunications (BUPT), Beijing, China, in 2020 and 2023, where he is currently pursuing the Ph.D degree with the State Key Laboratory of Networking and Switching Technology (SKT-NST), BUPT. His current research interests include integrated sensing and communication (ISAC), resource management, performance analysis and machine learning.
\end{IEEEbiography}

\begin{IEEEbiography}[{\includegraphics[width=1in,height=1.25in,clip,keepaspectratio]{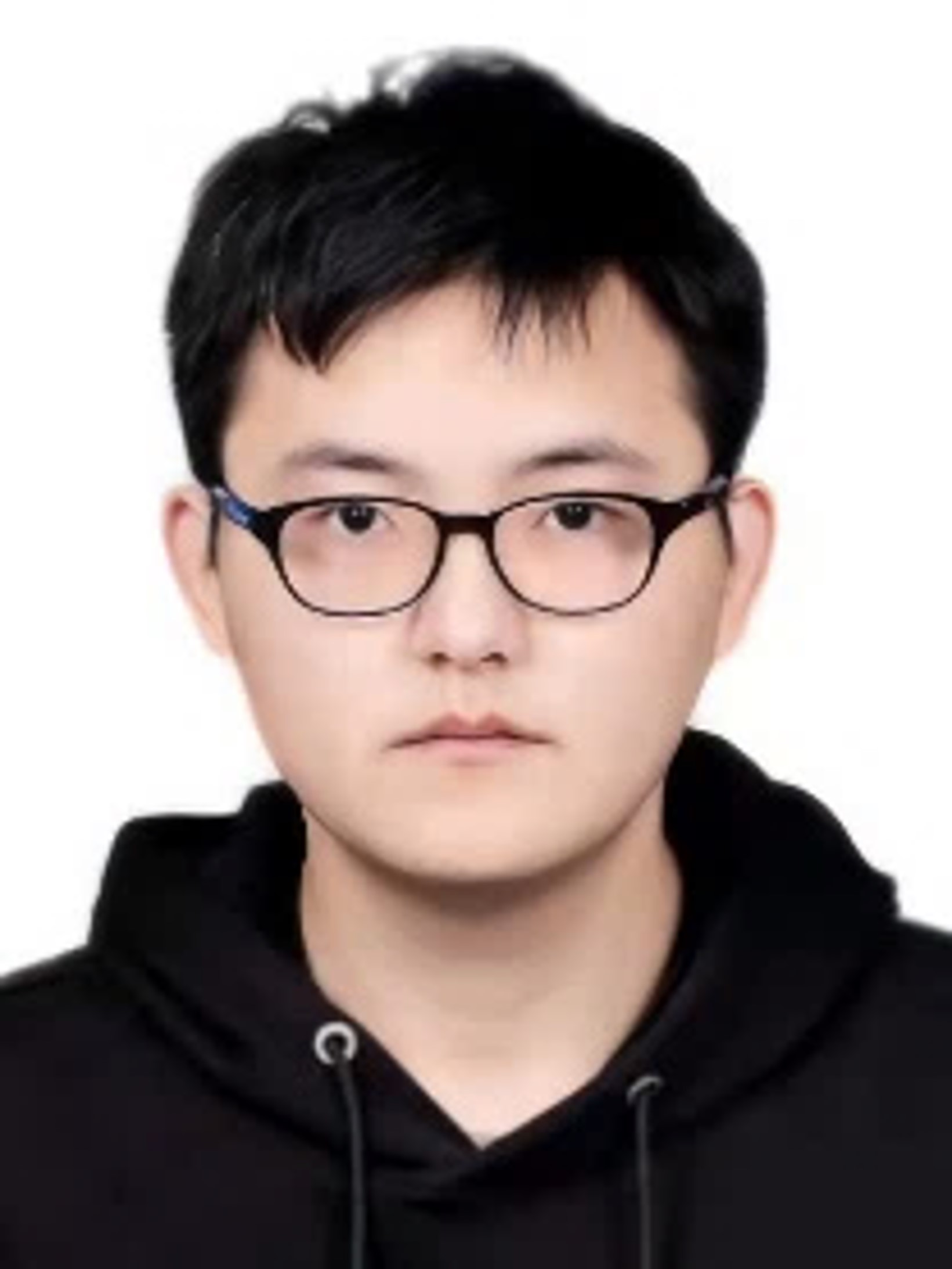}}]{Jiaen Zhou}
        (Student Member, IEEE) received the B.S. degree in communication engineering from the Beijing University of Posts and Telecommunications (BUPT), Beijing, China, in 2021, where he is currently pursuing the Ph.D degree with the State Key Laboratory of Networking and Switching Technology (SKT-NST), BUPT. His current research interests include integrated sensing and communication (ISAC), resource management and space-air-ground integrated networks.
\end{IEEEbiography}

\begin{IEEEbiography}[{\includegraphics[width=1in,height=1.25in,clip,keepaspectratio]{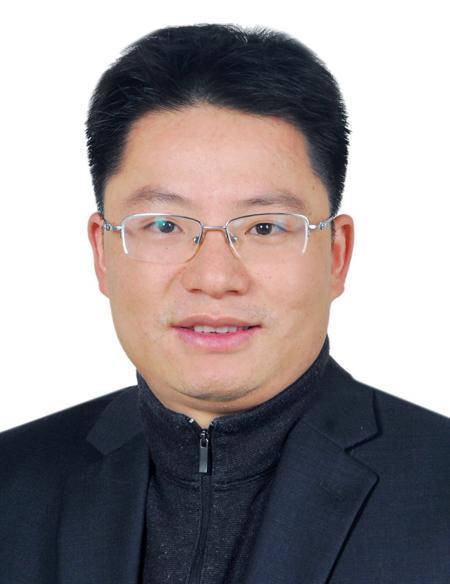}}]{Mugen Peng}
        (Fellow, IEEE) received the Ph.D. degree in communication and information systems from the Beijing University of Posts and Telecommunications (BUPT), Beijing, China, in 2005. Afterward, he joined BUPT, where he has been a Full Professor with the School of Information and Communication Engineering since 2012. During 2014 he was also an academic visiting fellow at Princeton University, NJ, USA. He leads a Research Group focusing on wireless transmission and networking technologies in BUPT. He has authored and coauthored over 100 refereed IEEE journal papers and over 300 conference proceeding papers. His main research areas include wireless communication theory, radio signal processing, cooperative communication, self-organization networking, heterogeneous networking, cloud communication, and Internet of Things. Prof. Peng was a recipient of the 2018 Heinrich Hertz Prize Paper Award; the 2014 IEEE ComSoc AP Outstanding Young Researcher Award; and the Best Paper Award in the ICC 2022, JCN 2016, IEEE WCNC 2015, IEEE GameNets 2014, IEEE CIT 2014, ICCTA 2011, IC-BNMT 2010, and IET CCWMC 2009. He is on the Editorial/Associate Editorial Board of the IEEE Communications Magazine, IEEE Access, the IEEE Internet of Things Journal, IET Communications, and China Communications. He is the Fellow of IEEE and IET.
\end{IEEEbiography}

\vfill
\end{document}